%% file: time-variant-backdoor.tex
\pdfoutput=1
\documentclass{article}

\usepackage{microtype}
\usepackage{graphicx}
\usepackage{subfigure}
\usepackage{booktabs} 
\usepackage{diagbox}
\usepackage{bm}
\usepackage{xspace}

\usepackage[dvipsnames]{xcolor}
\usepackage{url}
\PassOptionsToPackage{hyphens}{url}\usepackage{hyperref}

\usepackage[accepted]{icml2023}

\usepackage{amsmath}
\usepackage{amssymb}
\usepackage{mathtools}
\usepackage{amsthm}
\usepackage{multirow}
\usepackage{thmtools}
\usepackage{thm-restate}

\usepackage[capitalize,noabbrev]{cleveref}

\usepackage[textsize=tiny]{todonotes}

\input{defs.tex}

\icmltitlerunning{On the Permanence of Backdoors in Evolving Models}

\begin{document}

\twocolumn[
\icmltitle{On the Permanence of Backdoors in Evolving Models}




\begin{icmlauthorlist}
\icmlauthor{Huiying Li}{uchi}
\icmlauthor{Arjun Nitin Bhagoji}{uchi}
\icmlauthor{Yuxin Chen}{uchi}
\icmlauthor{Haitao Zheng}{uchi}
\icmlauthor{Ben Y. Zhao}{uchi}
\end{icmlauthorlist}

\icmlaffiliation{uchi}{University of Chicago}

\icmlcorrespondingauthor{Huiying Li}{huiyingli@cs.uchicago.edu}

\icmlkeywords{Machine Learning, ICML}

\vskip 0.3in
]



\printAffiliationsAndNotice{}  

\begin{abstract}

Existing research on training-time attacks for deep neural networks (DNNs), such as backdoors, largely assumes that models are static once trained, and  hidden backdoors  trained into models remain active indefinitely. In practice, models are rarely static but evolve continuously to address distribution drifts in the underlying data. 
This paper explores the behavior of backdoor attacks in {\bf \tvmodels}, whose model weights are continually updated via fine-tuning to adapt to data drifts. Our theoretical analysis shows how fine-tuning with fresh data progressively ``erases'' the injected backdoors, and our empirical study illustrates how quickly a time-varying model ``forgets'' backdoors under a variety of training and attack settings. We also show that novel fine-tuning strategies using smart learning rates can significantly accelerate backdoor forgetting.  Finally, we discuss the need for new backdoor defenses that target \tvmodels{} specifically.

\end{abstract}

\input{intro1}

\input{back}

\input{setup}
\input{theory}

\input{basic}

\input{attack}

\input{datashift}

\input{training}
\input{conclusion}


\bibliography{time-variant-backdoor}
\bibliographystyle{icml2023}

\newpage
\appendix

\input{appendix.tex}

\end{document}

%% file: defs.tex
\DeclareMathOperator*{\argmin}{arg\,min}

\theoremstyle{plain}
\newtheorem{theorem}{Theorem}[section]

\newtheorem{corollary}[theorem]{Corollary}
\theoremstyle{definition}
\newtheorem{definition}[theorem]{Definition}

\theoremstyle{remark}

\definecolor{ocean}{RGB}{78, 168, 245}

\newcommand{\mnist}{{\tt MNIST}}

\newcommand{\cifar}{{\tt CIFAR10}}

\newcommand{\para}[1]{{\vspace{-1pt} \noindent \textbf{#1}
		\hspace{3pt}}}

\definecolor{ForestGreen}{RGB}{24,155,24}

\newif\iffinal

\iffinal
    \newcommand{\yuxinnote}[1]{}
    \newcommand{\yuxin}[1]{}
\else
    
    \newcommand{\yuxinnote}[1]{\todo[fancyline,color=NavyBlue!40]{YC: #1}\xspace}
    \newcommand{\yuxin}[1]{\textcolor{NavyBlue}{[YC: #1]}}
\fi

\newcommand{\htedit}[1]{{\color{black} #1}}

\newcommand{\tvmodel}{time-varying model}
\newcommand{\TVModels}{Time-Varying Models}
\newcommand{\Tvmodels}{Time-varying models}
\newcommand{\tvmodels}{time-varying models}
\newcommand{\tabor}{{\it\tt TABOR}}
\newcommand{\nc}[1]{{\it\tt Neural Cleanse}}
\newcommand{\fp}{{\it\tt Fine-Pruning}}

\newcommand{\btheta}{\ensuremath{{\bm{\theta}}\xspace}}

\newcommand{\patch}{\textsc{patch}}

%% file: intro1.tex
\section{Introduction}
\label{sec:intro}


Deep neural networks (DNNs) are vulnerable
to backdoor attacks, where attackers corrupt training data in order to introduce incorrect model behavior on inputs with specific
characteristics~\cite{chen2017targeted, gu2017badnets, liu2018trojaning}.
Backdoors are challenging to detect and mitigate, and are considered the most worrisome attacks by industry practitioners~\cite{kumar2020adversarial}.  Existing proposals to detect and 
mitigate backdoors (e.g.,~\cite{gao2019strip,  finepruning, liu2019abs, tran2018spectral, wang2019neural}) fall short when examined in a variety of attack settings, including transfer learning~\cite{yao2019latent}, federated
learning~\cite{pmlr-v97-bhagoji19a,bagdasaryan2020backdoor, wang2020attack, xie2019dba} and
physical attack scenarios~\cite{lin2020composite, wenger2021backdoor}.

Existing works on backdoor attacks and defenses share a common assumption: the model, once trained, will never change. As such, an injected backdoor will stay in the model permanently. In reality, however, deployed models rarely if ever remain static, but continuously evolve to incorporate more labeled data or adapt to drifts in the underlying data distribution~\cite{kumar2020understanding,
  nahar2022collaboration,retrain1}. As the target distribution shifts, a static model will continue to drop in performance over time. 


In this work, we study the permanence of backdoor attacks on {\bf \tvmodels}, models that are periodically updated with new training data. 
Specifically, we consider \tvmodels{} updated regularly via fine-tuning\footnote{While there are multiple ways of updating \tvmodels{}, including transfer learning/fine-tuning, online (incremental) learning and domain adaptation~\cite{hoi2021online,
hoi2014libol, read2012batch, sahoo2018online, yoon2018lifelong}, we choose to focus on fine-tuning due to its generality and wide adoption. We leave the study on other model updating methods to future work.}, the most widely used method for addressing data distribution drifts~\cite{yosinski2014transferable, weiss2016survey}. 
We note that existing studies on backdoors have shown that fine-tuning is ineffective as a defense against backdoors, i.e., it fails to remove backdoors from a static model~\cite{finepruning}.  In contrast, our work asks a different question: 

''{\em For \tvmodels{} that are periodically updated via fine-tuning, how long can backdoors survive before they are removed or become ineffective?}''

Given the lack of robust defenses against backdoor attacks, it is important to first understand the temporal behavior of backdoors as the models evolve themselves. 





\para{Our Contributions.} 
In this paper, we report results from a comprehensive study on the permanence of backdoors in \tvmodels{} periodically updated by fine-tuning.  Our work makes four key contributions: 

(1) We present the first theoretical analysis of backdoor 
and model behavior during periodic fine-tuning. Our analysis shows that with sufficient training, fine-tuning can remove backdoors inside corrupted models, while more training updates and larger learning rates can accelerate the forgetting process. We also show, for the first time, that the number of iterations required to fine-tune the model to recover benign behavior and proportion of poisoned training data, thus theoretically linking the model convergence and poisoning ratios.


(2) We take an empirical approach to study the behavior of backdoors in \tvmodels{} with complex and realistic training dynamics. We quantify   how embedded backdoors gradually degrade in \tvmodels{}  until they are ``forgotten'' and become ineffective. 
We define a new metric called {\bf backdoor survivability} on \tvmodels{}, and
explore the impact of poison ratio, trigger evasiveness and benign data shifts
on backdoor survivability.

(3) Leveraging insights from our theoretical analysis, we show that well-chosen training strategies can significantly reduce the
backdoor survivability with negligible overhead. 

(4) We discuss the
need of new backdoor defenses under the time-varying setting, because most existing defense
mechanisms cannot be directly adapted to the time-varying setting.

%% file: back.tex
\vspace{-0.05in}
\section{Background and Related Works}\label{sec:back}

In this section, we provide the necessary background on backdoor attacks and time-varying DNN models.

\vspace{-0.05in}
\subsection{Backdoor Attacks}\label{subsec: backdoor_attacks}
\vspace{-0.05in}

Backdoors are the most common and effective form of data poisoning attacks,
and also the most concerning type of attacks against machine learning systems~\cite{kumar2020adversarial}.
The initial backdoor attacks assume the attacker only has
access to the training data~\cite{gu2017badnets, chen2017targeted}. Later work
proposes a rich set of attack variants which consider stronger attackers with white-box access to the models and training process~\cite{liu2018trojaning, yao2019latent, nguyen2020input, tang2020embarrassingly, liu2020reflection, pang2020tale}.
 Table~\ref{tab:attacks} in Appendix summarizes the SOTA attacks and their threat models. 

Backdoor attacks are also stealthy and hard to
detect. Unlike traditional poisoning attacks~\cite{biggio2012poisoning,
  xiao2012adversarial, steinhardt2017certified}, backdoors typically do not
affect normal model accuracy, and only cause misclassification on
inputs containing attack-specific triggers. While significant efforts were
spent on detecting and removing backdoors on DNN models
(e.g.,~\cite{gao2019strip, tran2018spectral, finepruning, wang2019neural, liu2019abs,
qiao2019defending}), they all face major limitations~\cite{gao2020backdoor, veldanda2020evaluating}.

\vspace{-0.05in}
\subsection{\TVModels}\label{subsec: tv_models}
\vspace{-0.05in}
Machine learning models are usually trained based on the assumption that the
distribution of training and test data is identical. In practice, this is often not true. Test data can come from distributions
that have drifted away from the training data distribution, and this can
significantly affect model accuracy~\cite{patel2015visual, wang2018deep}. 

It is common practice to
update a deployed model over time in order to handle distribution drifts~\cite{MLops, tajbakhsh2016convolutional,kumar2020understanding, nahar2022collaboration}.
Numerous methods have been proposed to address data distribution drifts, including transfer learning~\cite{torrey2010transfer, yosinski2014transferable, weiss2016survey, zhuang2020comprehensive}, online (incremental) learning~\cite{read2012batch, sahoo2018online, hoi2021online, chen2021active} and domain adaptation~\cite{ganin2015unsupervised, sankaranarayanan2018learning, liu2021cycle}.

\para{Poisoning attacks on \tvmodels.} 
Recent work has studied non-backdoor poisoning attacks on
\tvmodels, specifically online learning
\cite{wang2018data,wang2019investigation,pang2021accumulative}. These attacks aim at lowering normal model accuracy. 


\vspace{-0.05in}
\subsection{Model Fine-Tuning}\label{sebsec: fine-tuning}
\vspace{-0.05in}

\para{Fine-tuning in transfer learning.} Fine-tuning is widely used in the context of transfer learning~\cite{yosinski2014transferable},
where the model trainer starts from a pretrained (teacher) model and updates the model using fresh training data at a small learning rate. 
Fine-tuning is significantly faster than training-from-scratch, e.g., training AlexNet via fine-tuning takes an hour~\cite{finepruning} but $>6$ days when training-from-scratch~\cite{iandola2016firecaffe}.
 
\para{Fine-tuning as a defense against backdoor attacks.} 
Previous studies have shown that fine-tuning is an ineffective and unstable defense against backdoor attacks~\cite{finepruning, pang2022trojanzoo}, although these studies only considered static models. 

To the best of our knowledge, there is no prior study on the behavior/effectiveness of backdoors in \tvmodels, where models are periodically updated using fine-tuning and fresh data to handle data drifts.  This motivated our study.

%% file: setup.tex
\vspace{-0.05in}
\section{Definitions and Threat Model}
\label{sec:setup}
We introduce definitions of \tvmodel{} and backdoor attacks, as well as the threat model for our study. We focus on image-based classification tasks.

\vspace{-0.05in}
\subsection{Definitions}
\vspace{-0.05in}
\para{\Tvmodels.}
We refer to models that change over time as \tvmodels{} and define them below: 

\begin{definition}[\Tvmodels]
	A \tvmodel{} $F$, observed between time 0 and $t$,  is a sequence of
        models $\{F_{\btheta_i}\}_{i=0}^{t}$ trained from sequentially available data
        $\{D_i\}_{i=0}^{t}$ such that after the $i^{th}$ update, the model
        $F_{\btheta_i}=\mathbf{g}(\{D_k\}_{k \leq i}, \{F_k\}_{k < i})$. Here
        $\mathbf{g}(\cdot)$ is any training function. The new training data
        available at the $i^{th}$ update ($D_i$) is drawn from a distribution
        $\mathcal{D}_i$.
\vspace{-0.05in}
\end{definition}

We note that model updates can come from different learning paradigms such
as transfer learning, online (incremental) learning, and domain adaptation. 
In this paper, we focus on model fine-tuning (transfer learning), 
as it is the most basic and common way to update \tvmodels.


\para{Backdoor attacks.}
We define a backdoor attack by its trigger with a patch function $\patch(\cdot)$ on
the inputs, a target label $y_t$, and a poison ratio $(\alpha)$
where $\alpha$ defines the fraction of training data modified by the attacker.
By injecting a collection of poisoned training data $\{(\patch(x), y_t) | (x,y) \sim D, y\ne y_t\}$,
the attacker is able to inject trigger-specific classification rules into the trained model by minimizing the poison loss:
\begin{align}%
\begin{split}
\vspace{-0.1in}
	L_p (\btheta)&=	\mathbb{E}_{(x,y) \sim D| y\ne y_t}[\ell_p((x,y),F_{\btheta})] \\
	& = \mathbb{E}_{(x,y) \sim D| y\ne y_t}[\ell(y_t, F_\btheta(\patch(x)))]
 \vspace{-0.05in}
\end{split}
\end{align}


\vspace{-0.05in}
\subsection{Threat Model}\label{subsec:threat}
\vspace{-0.05in}

We use the standard threat model of backdoor attacks
where the attacker is only able to control or modify a fraction of the model training data,
and has no knowledge of the model architecture, weights, or training
hyperparameters and no control of the training process during initial model training and subsequent model updates. 

We focus on {\bf one-shot poisoning} where the attacker is only able to poison
the training data once. Without loss of generality, we consider
two cases: (i) the attacker poisons $D_0$ and injects a backdoor
into $F_{\btheta_0}$,  the initial model prior to deployment;  and (ii) the
attacker poisons the data used by a model update $D_i$ and thus
injects a backdoor into $F_{\btheta_i}$ (\htedit{via fine-tuning}). Without loss
of generality, we consider $i=1$.

We note that more powerful attackers can always continuously poison the training data used by subsequent model updates ($\{D_i\}_{i>0}$). Such persistent poisoning increases backdoor survivability (detailed results in Appendix~\ref{subsec:persist}). 





%% file: theory.tex
\vspace{-0.05in}
\section{
Theoretical Analysis of Backdoor Attacks during Fine-tuning}
\label{sec:theory}
In this section, we first state a general proposition on the need of model fine-tuning in the presence of distribution drifts. We then show that for strongly convex loss functions, backdoors can be removed through sufficient fine-tuning with data from the original distribution. Larger learning rates can accelerate this process. However, the number of iterations required is influenced by the proportion of poisoned training data.

It is a well-established result that when the test distribution does not match the training distribution, the model has to be fine-tuned or re-trained for better generalization performance. Formally, the need for model fine-tuning is reflected in the following proposition \cite{ben2010theory,kumar2020understanding}.

\begin{restatable}[Need for fine-tuning with distribution drift]{proposition}{finetuning}
	Let $\btheta^*_t = \argmin_\btheta  L_{D_t}(\btheta)$ and $\btheta^*_{t-1} = \argmin_\btheta L_{D_{t-1}}$.  Then, 
 $L_{D_t}(\btheta^*_{t-1}) \geq L_{D_t}(\btheta^*_{t})$. 
	\label{thm:finetuning_general}
\end{restatable}

In the following, we analyze the impact of model fine-tuning on backdoor attacks, under the special setting of one-shot poisoning with \textit{strongly convex} loss functions.

\paragraph{Theoretical setting.} We assume the classifiers $\btheta$ lie in a convex set $\Theta$ are norm-bounded by $B$. Both the loss functions $\ell$ and $\ell_p$ are strongly convex and $\gamma$-smooth with parameters $(\sigma_b,\gamma_b)$ and $(\sigma_p,\gamma_p)$ respectively. A function $f$ is $\sigma-$strongly convex if $\left(  \nabla f(x) -\nabla f(y)  \right)^\intercal (x-y) \geq \sigma \| x-y \|^2 $. A differentiable function is $\gamma$-smooth if $\| \nabla f(x) - \nabla f(y) \| \leq L \|x-y\|$. We define $L_D(\btheta) = \mathbb{E}_{(x,y) \sim D}[\ell((x,y),F_{\btheta})]$. The attacker trains a poisoned classifier $\btheta_{\text{mix}}$ using a composite loss function $L_{\text{mix}}($\btheta$) = \alpha L_p($\btheta$) + (1-\alpha) L_{D_0}($\btheta$)$, where $\alpha$ represents the fraction of poisoned data. Proofs for the results in this section are deferred to Appendix~\ref{sec:app_theory}.

We first provide an upper bound on the distance between the poisoned model $\btheta_{\text{mix}}$ and the one trained on the benign distribution, i.e., $\btheta_0^* = \argmin_\btheta L_{D_0}($\btheta$)$.




\begin{restatable}
{lemma}{poisonslowdown}
 $\| \btheta_{\text{mix}} - \btheta_0^* \|\leq \frac{\alpha \lVert \nabla L_p(\btheta_0^*) \rVert}{\alpha \sigma_p + (1-\alpha) \sigma_b}$. 
\label{lemma:poison_slowdown}
\end{restatable}
Note that the upper bound in Lemma~\ref{lemma:poison_slowdown} is a monotonically increasing function of $\alpha$. The following theorem establishes the convergence rate of stochastic gradient descent for model fine-tuning, assuming that the model is initialized with the poisoned model $\btheta_{\text{mix}} = \argmin_{\btheta} L_{\text{mix}}(\btheta)$. 
\begin{restatable}[Effectiveness of model fine-tuning for backdoor removal]{theorem}{backdoorremoval}
Fine-tuning $\btheta_{\text{mix}}$ with stochastic gradient descent with a learning rate of $\eta=\frac{1}{\sigma_b t}$ on 
$D_0$ leads to a classifier $\hat{\btheta}$ which is $\epsilon$-close to $\btheta_0^* = \argmin L_{D_0}(\btheta)$ in $ \frac{\alpha \gamma_p}{\epsilon (\alpha \sigma_p + (1-\alpha) \sigma_b)}$ iterations.
\label{thm:backdoorremoval}
\end{restatable}

According to Theorem~\ref{thm:backdoorremoval}, fine-tuning with clean data will
remove the backdoor from the poisoned model with sufficient fine-tuning. 
Intuitively, the model recovers more from the backdoor attack (i.e. gets closer to the benign minimum) with more iterations. However, the stronger the poisoning is, the greater the number of iterations that will be needed to converge to the benign minimum.

\begin{corollary}[Large initial learning rates speed up backdoor removal]
  Using an adaptive learning rate of $\eta_t=\frac{1}{\sigma_b t}$ leads to the fastest rate of convergence to $\btheta^*_0$.
    \label{coro:lr}
\end{corollary}

 Given limited training
resources in reality, we could use a smarter learning rate during the training to
speed up the backdoor removal process as shown in Corollary~\ref{coro:lr}.

%% file: basic.tex
\section{Backdoor Survivability against Periodic Model Fine-Tuning}
\label{sec:basic}
\vspace{-0.05in}

Following our theoretical analysis in \S\ref{sec:theory},
we now conduct an empirical study examining the behavior of backdoors
in \tvmodels{} with more complex and realistic training dynamics.
We investigate backdoor attacks on state-of-the-art DNNs which are
updated through fine-tuning over time to handle data distribution shifts.
To the best of our knowledge, this is the first study to comprehensively
investigate the effects of periodically updating models through fine-tuning on backdoor attacks.

We formulate our empirical study to examine the ``survivability'' of a
successfully injected backdoor against subsequent model updates, and how
backdoor survivability is affected by different attack configurations, 
data drift behaviors, and model update strategies.  We believe these
results offer a more in-depth understanding of the behavior of
backdoor attacks against production models, in terms of both attack
overhead and damage to the models.
Later in \S\ref{sec:training} we leverage these insights to build model training/updating strategies that
further reduce backdoor survivability. 

In the following, we first introduce the experimental setup (\S\ref{subsec: exp_setup}), and a formal definition of backdoor survivability and some initial observations (\S\ref{subsec:surv}).   We present the
detailed results in \S\ref{sec:attack} -- \S\ref{sec:shift}.

\vspace{-0.05in}
\subsection{Experimental Setup}\label{subsec: exp_setup}
\vspace{-0.05in}
We now describe the configuration of our empirical study to examine backdoor
survivability.

\para{Datasets.} To build a controlled pipeline for our dynamic
data environment, we propose two semi-synthetic datasets for image classification:
\mnist{}~\cite{lecun1998gradient} and \cifar{}~\cite{krizhevsky2009learning}.
For each task,  we randomly split its training data into two halves. The first half
is assigned as $D_0$ and used to train the initial model $F_0$.
The second half and the test data are used to emulate the dynamic
data environment, where both the training and test data distributions vary over
time.

In this work, we produce parameterized data dynamics by
applying  image transformations progressively (i.e.,
changing angle for \mnist{} and hue, brightness, saturation for \cifar{}\footnote{We do
not change angles for \cifar{} given rotating the images from \cifar{} will lose
information from the original images and we do not change hue, brightness, saturation for \mnist{}
since \mnist{} is mostly black and white and applying these transformation does little change on it.}),
similar to the method used by~\cite{kumar2020understanding} to study gradual domain
adaptation. We note that these semi-synthetic datasets
cannot capture the exact data distribution in the wild, but the drifts are controllable
and we can launch fine-grained empirical study on them.

In our experiments, the default data drift $p$ between any two
consecutive updates is caused by changing the angle of the current images by
a factor of $4^\circ$ on \mnist{} and the hue of the current images by
a factor of $0.02$ on \cifar{}.  These drifts, if not addressed, cause significant
degradation to model performance (e.g., the classification accuracy
drops from 99\% to {\bf 46\%} after 15 rounds of angle change for \mnist{} and
 92\% to {\bf 74\%} after 15 rounds of hue change for \cifar{} as shown in Figure~\ref{fig:acc_drop} in Appendix~\ref{subsec:acc_drop}).
 This is consistent with Proposition~\ref{thm:finetuning_general} in our theoretical analysis. More details on distribution drifts are listed in Appendix~\ref{sec:app_setup}.

\para{DNN models and updates.}  We present the results on
the ResNet-9~\cite{ilyas2022datamodels} model architecture.
We also verify our results on ResNet18~\cite{he2016deep} and
DenseNet121~\cite{huang2017densely} and find that they produce the
same trends on backdoor survivability.  For brevity, we only present the results
for ResNet-9, which is optimized for fast training.
For both \cifar{} and \mnist{}, we train the initial model $F_0$ to
reach a high normal classification accuracy: $92\%$ for \cifar{} and
$99\%$ for \mnist.

We consider fine-tuning as the model update strategy due to
its efficiency and practicality (as discussed in \S\ref{sebsec: fine-tuning}).  
At the $i^{th} (i>0)$ update, we fine-tune model $F_{i-1}$
using new training data $D_i$ to produce $F_i$, by applying
stochastic gradient descent (SGD) with weight decay and
momentum. We set the default learning rate to $0.01$. 
More details of model training and updating are listed in Appendix~\ref{sec:app_setup}.

\para{Attack configuration.}  We consider 3 attacks which have the same
threat model as ours: Badnets~\cite{gu2017badnets},
Blend~\cite{chen2017targeted}, and Wanet~\cite{nguyen2021wanet}. By default, the attacker is
able to poison 10\% of the training data $D_i$.
Figure~\ref{fig:example} gives an example of backdoor data for the 3 different attacks
on \cifar{}.  We conduct detailed experiments to
explore the impact of poison ratio and trigger stealthiness
in \S\ref{sec:attack}. For all our experiments (except the experiments 
varying poison ratios), we ensure that
the backdoor, when injected into the initial model ($F_0$), is
effective, i.e., its attack success rate is at least
$99\%$ without affecting the model's normal accuracy.

For each experiment, we first generate five different trained $F_0$ instances at random. Then for each
instance, we run 15 model updates with even data distribution drifts with step $p$
after the inital training.  We report the average results over the 5
instances.  We run all the experiments using the FFCV library~\cite{leclerc2022ffcv}
on an NVIDIA TITAN RTX GPU with 24,576MB GPU memory.

\begin{figure}[t]
  \centering
  \includegraphics[width=0.9\linewidth]{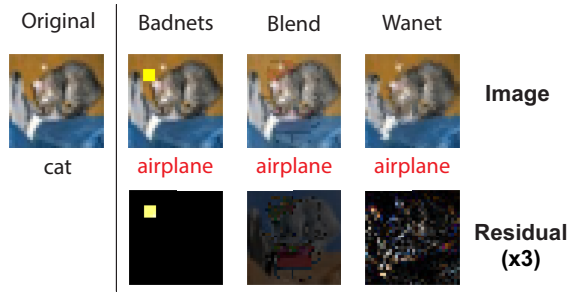}
  \caption{Examples of benign and poisoned training data and their magnified ($\times 3$) residual map
   with labels generated by the three attacks (Badnets~\cite{gu2017badnets}, Blend~\cite{chen2017targeted}, Wanet~\cite{nguyen2021wanet}) on \cifar.}
  \label{fig:example}\vspace{-0.1in}
\end{figure}

\begin{figure}[t]
\centering
  \subfigure[One-shot poison in $D_0$ (Model Update Event 0)]{
  \includegraphics[width=.95\linewidth]{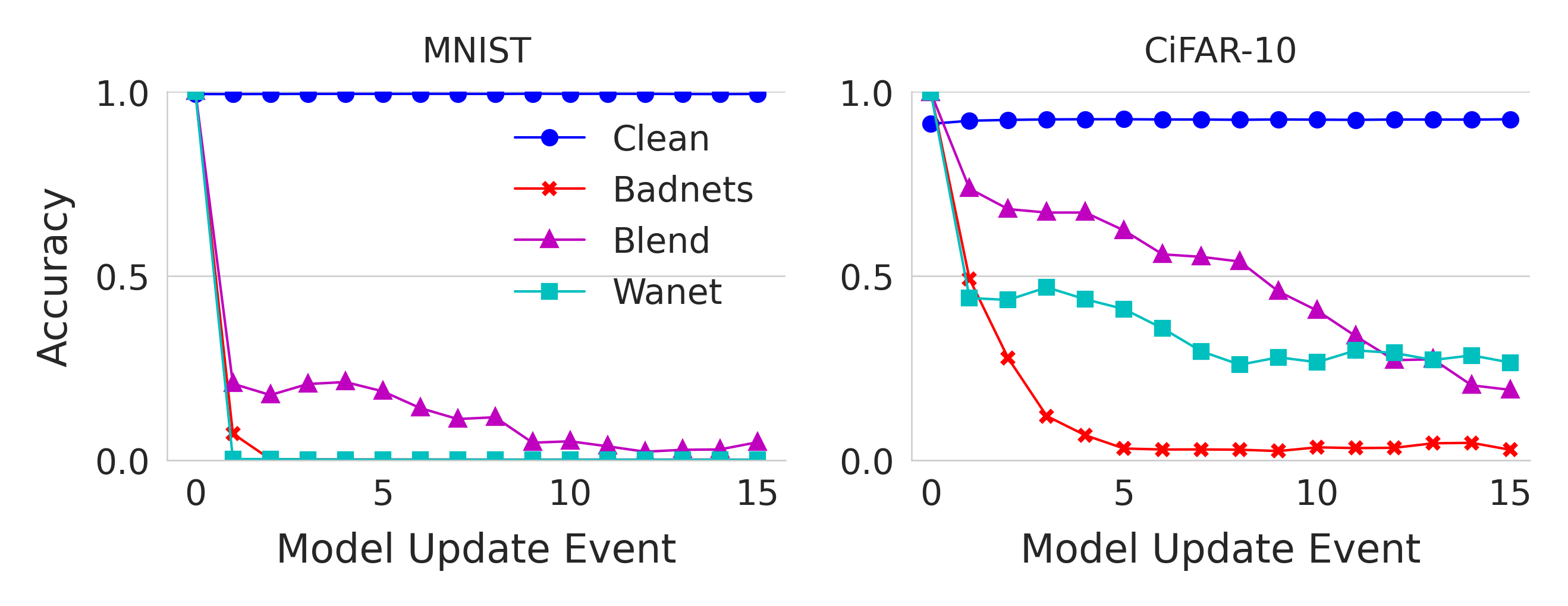}
  \label{fig:oneshot_D0}
  }
  \subfigure[One-shot poison in $D_1$ (Model Update Event 1)]{
  \includegraphics[width=.95\linewidth]{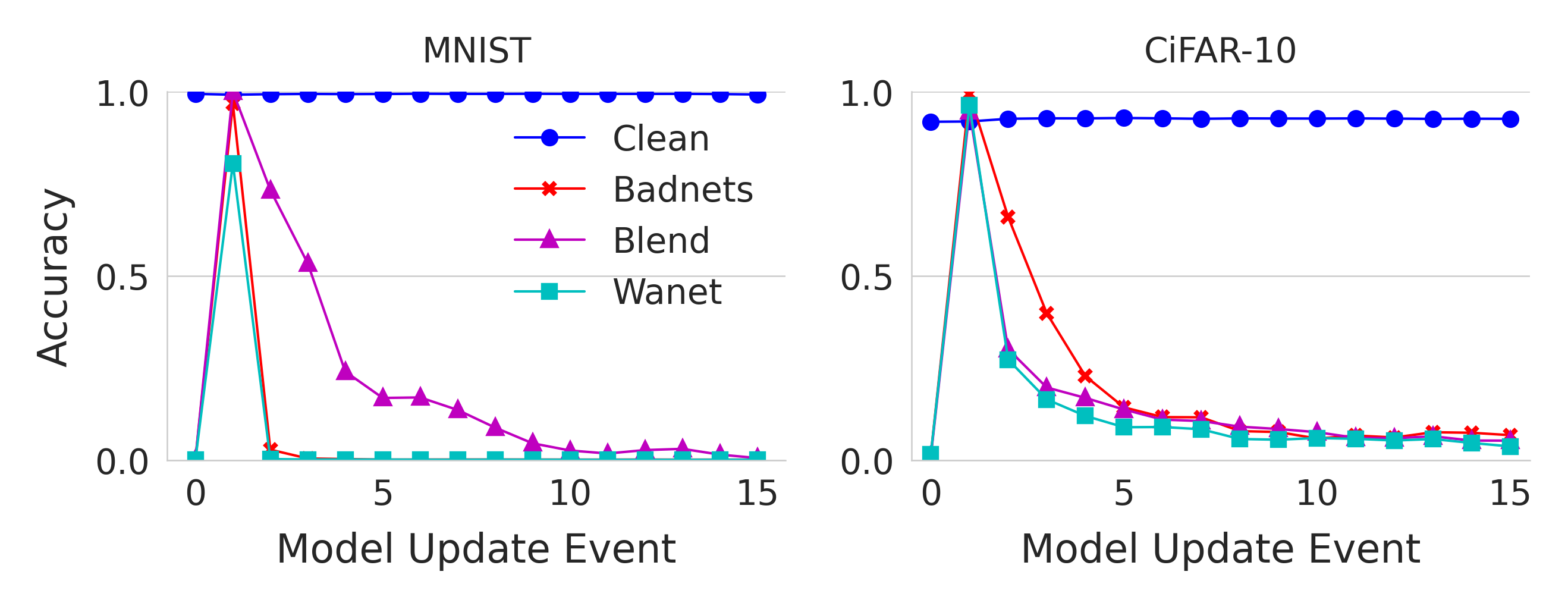}
  \label{fig:oneshot_D1}
  }\vspace{-0.05in}
\caption{Normal accuracy and attack success rate for one shot poisoning
using different attack methods on \mnist{} and \cifar{}.
`Clean' represents the average normal accuracy (averaged on 15 models,
5 for each attack method), `Badnets', `Blend' and `Wanet' represent the attack
success rate for each attack method.}
  \vspace{-0.1in}
\label{fig:oneshot}
\end{figure}

\vspace{-0.05in}
\subsection{Defining Backdoor Survivability}
\label{subsec:surv}
\vspace{-0.05in}
We define the survivability of a backdoor as the ``window of
vulnerability'' it produces on the target \tvmodels{} after the poisoning stops:

\begin{definition}[$\gamma$-survivability of a backdoor attack]
  The $\gamma$-survivability of a backdoor is the maximum number of
  subsequent model updates once the poisoning stops, during which the attack
  success rate on the target model remains above a  threshold $\gamma$.
\end{definition}\vspace{-0.06in}
Here $\gamma$  can vary depending on the definition of
``model vulnerability''. Researchers can choose different $\gamma$ according to
the context and application. In this paper, we choose $\gamma=0.5$.
We report the $0.5$-survivability for the 14 subsequent model updates after a backdoor is injected
for equal comparison between injection to $D_0$ and $D_1$ ($\max(\gamma-\text{survivability}) = 14$).


\para{Initial observations.}  For one-shot poisoning attacks, our
initial hypothesis is that the attack success rate should always decrease
over time, because as more benign samples are used to train the model,
the influence of the poisoned samples should reduce which aligns with our
theoretical results (Theorem~\ref{thm:backdoorremoval}). Figure~\ref{fig:oneshot}
confirms that for all three attacks on both datasets, the backdoor gradually
degrades when the model is fine-tuned to learn new (benign) data,
once the poison stops (both for poisoning $D_0$ and $D_1$).
At the same time, we can see that the normal accuracy stays at the same level during model fine-tuning with
data distribution drifts. Thus, we omit the normal accuracy and only report backdoor
survivability in the following sections for direct comparison over different attack
and training settings.

%% file: attack.tex
\vspace{-0.05in}
\subsection{Impact of Attack Configurations}\vspace{-0.05in}
\label{sec:attack}
We now present the results of $\gamma$-survivability under different
attack strategies and configurations.

As discussed in \S\ref{subsec:threat}, we consider two attack scenarios based
on how backdoors are injected into the target model. The attacker can (i)
poison $D_0$ to inject a backdoor to the original model $F_0$ via training
from scratch (or full training); or (ii) poison $D_1$ to inject a backdoor to
$F_1$ via fine-tuning.  While both can inject backdoors successfully\footnote{
As shown in Figure~\ref{fig:oneshot}, the attack success rate is over 98\% in
most cases except Wanet achieves 80.5\% attack success rate when injecting into $D_1$ on \mnist{} given the
attack is harder in order to make the trigger more invisible.
}, we want
to understand how injection methods affect the backdoors' survivability
against model updates.  We are also interested in if/how backdoor
survivability is affected by attacker-side parameters, including poison
ratio and trigger evasiveness.

We make 4 key observations from our results.

\para{Fine-tuning based model updates gradually remove one-shot backdoors.}
Figure~\ref{fig:oneshot} shows that using either injection method,
one-shot backdoor success rate degrades with additional model updates. After
a poisoned model learns backdoors as hidden classification features, each
updated version gradually forgets these features over time, if model updates
are benign and there are no poisoned samples to reinforce the model's memory.
While they were not designed to address backdoors, fine-tuning based
model updates {\em naturally} degrade one-shot backdoors over time.
This finding is consistent with our analytical study.
%


\para{Backdoors injected via full training usually survive longer than those injected
  via fine-tuning.}  Another interesting finding is that backdoors embedded
in $F_0$ usually carry more ``strength'' and survive longer against model updates
compared to those injected into $F_1$.  As shown in Figure~\ref{fig:oneshot_pr},
the backdoor survivability for poisoning $D_0$ (Figure~\ref{fig:oneshot_D0_pr}) is
generally higher than poisoning $D_1$ (Figure~\ref{fig:oneshot_D1_pr}) with
the same poison ratio. This is
likely because backdoors injected via full training embed the hidden features
more broadly into the model, making them harder to ``remove'' by subsequent
fine-tuning.  On the other hand, poisoning $D_0$ is generally more
challenging, because there is much more opportunity to apply backdoor
defenses and detectors on $F_0$ before its deployment.

\para{Higher poison ratio may degrade backdoor survivability.}
Existing works have shown that higher poison ratios embed a
backdoor with a higher success rate. In our theoretical analysis (Theorem~\ref{thm:backdoorremoval}), we also find that larger poison ratios lead to more iterations of fine-tuning for strongly convex loss functions. Surprisingly, larger poison ratio does not always improve backdoor survivability against model updates on DNN models.
Our findings indicate that as the proportion of poisoned data increases from zero but remains relatively low, the ability of the backdoor to survive improves. However, as the proportion of poisoned data surpasses a certain threshold, the backdoor survivability decreases and eventually reaches zero across all attack types and datasets. This suggests that a higher proportion of poisoned data can actually hasten the process of backdoor forgetting in DNN models. This also indicates the possible presence of 
a more complex theoretical relationship between backdoors and feature forgetting for DNN models, since they use more complicated loss functions. 
We leave this to future work.

\begin{figure}[t]
\centering
  \subfigure[One-shot poison in $D_0$]{
  \vspace{-0.1in}
  \includegraphics[width=\linewidth]{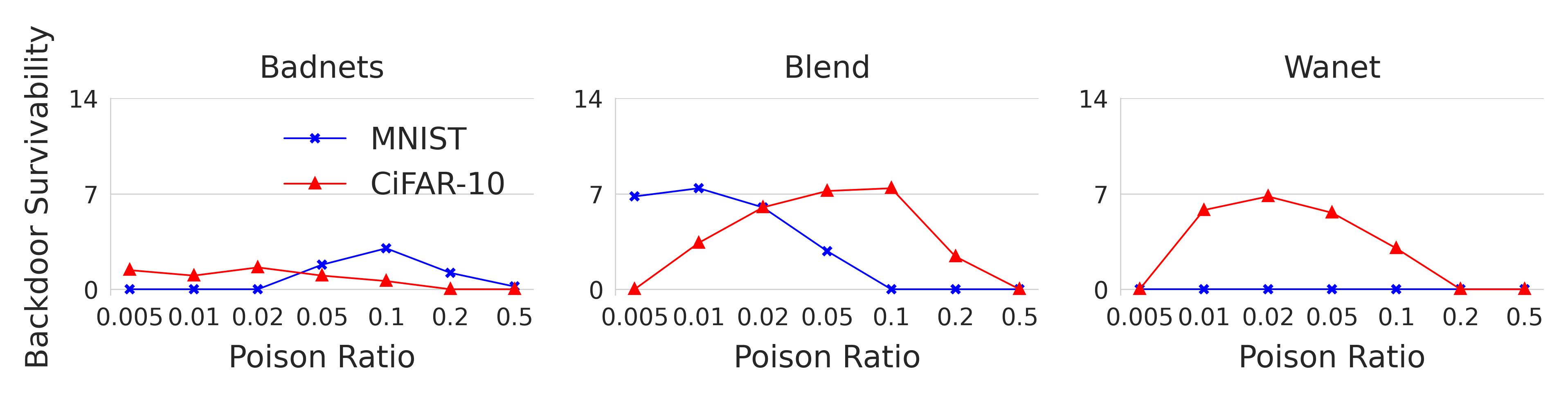}
  \vspace{-0.1in}
  \label{fig:oneshot_D0_pr}
  }
  \vspace{-0.1in}
  \subfigure[One-shot poison in $D_1$]{
  \vspace{-0.1in}
  \includegraphics[width=\linewidth]{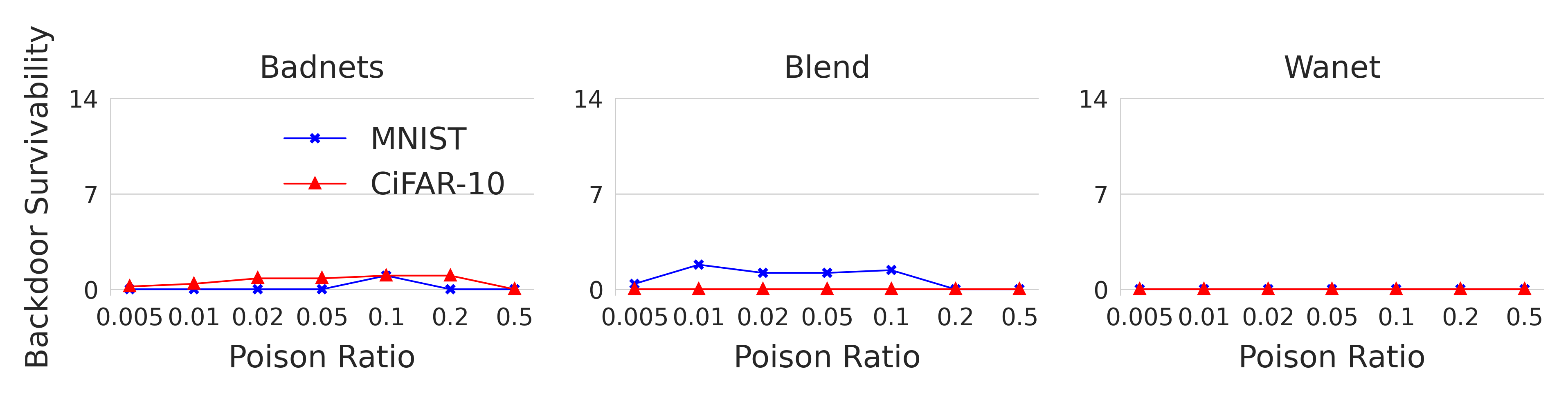}
  \vspace{-0.1in}
  \label{fig:oneshot_D1_pr}
  }
  \vspace{-0.1in}
\caption{Backdoor survivability for one shot poisoning
using different poison ratios for the 3 attack methods on \mnist{} and \cifar{}.
The results are averaged on 5 instances.}
\label{fig:oneshot_pr}
\vspace{-0.1in}
\end{figure}

\begin{figure}[t]
\centering
  \subfigure[One-shot poison in $D_0$]{
  \vspace{-0.1in}
  \includegraphics[width=.95\linewidth]{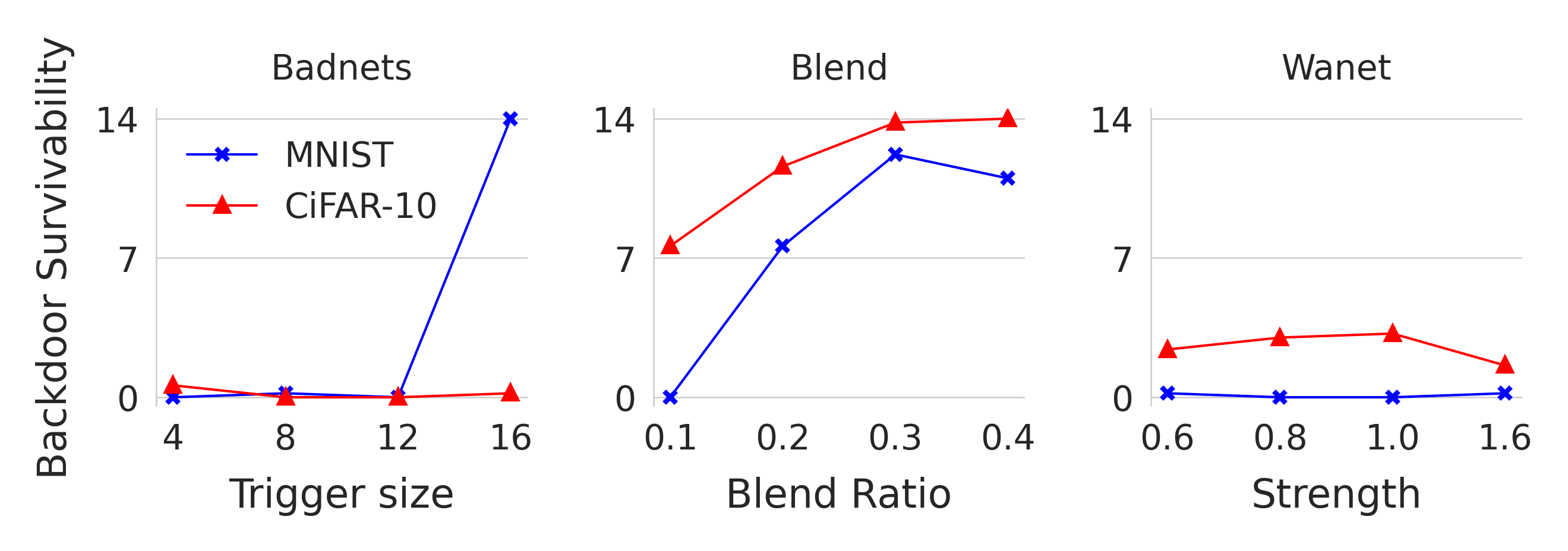}
  \vspace{-0.1in}
  \label{fig:oneshot_D0_te}
  \vspace{-0.1in}
  }
  \subfigure[One-shot poison in $D_1$]{
  \vspace{-0.1in}
  \includegraphics[width=.95\linewidth]{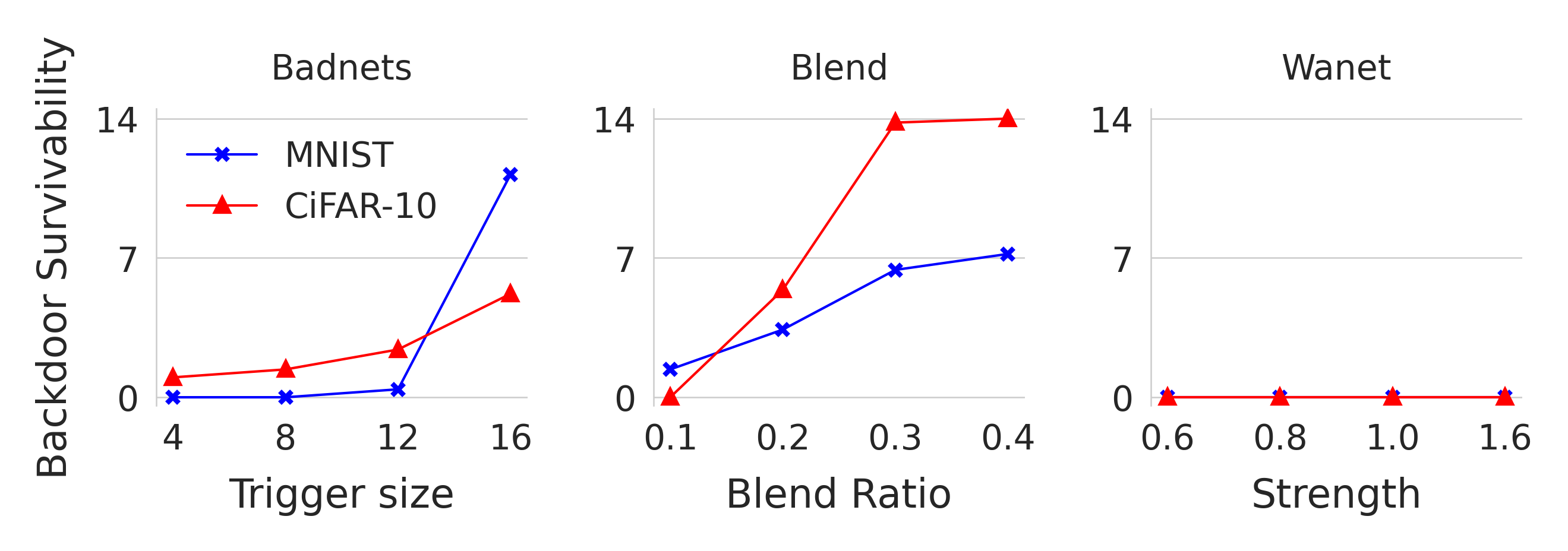}
  \vspace{-0.1in}
  \label{fig:oneshot_D1_te}
  \vspace{-0.1in}
  }\vspace{-0.1in}
\caption{Average backdoor survivability for one shot poisoning using different triggers
for the 3 attack methods on \mnist{} and \cifar{}.}
\vspace{-0.15in}
\label{fig:te}
\end{figure}


\para{There exists a tradeoff between trigger evasiveness and backdoor survivability.}
Our results indicates that when reducing the trigger evasiveness, which makes the
trigger more noticeable, the backdoor survivability increases.
As shown in Figure~\ref{fig:te}, when we increase the size of the trigger for Badnets,
the blend ratio of the trigger for Blend or increase the strength factor for Wanet,
which all make the backdoor trigger more visible, the backdoor survivability
increases correspondingly. As the most invisible attack, Wanet also has the smallest 
backdoor survivability. This is likely because larger/more obvious triggers introduce stronger
deviations on the model's decision manifolds.  However, using larger triggers
often means lower attack stealthiness and more visible perturbations.
In this case, physical backdoors using real world objects as triggers
might achieve a sweet spot of stealth and survivability~\cite{wenger2021backdoor}.

%% file: datashift.tex
\vspace{-0.05in}
\subsection{Impact of Data Distribution Drift}
\label{sec:shift}
\vspace{-0.05in}
As discussed in \S\ref{subsec: exp_setup}, we emulate a dynamic data
environment by introducing parameterized distribution shifts using
image transformations.  So far our results assume the default data distribution
shift (i.e., changing angle by $4^\circ$ for \mnist{} and changing hue by a factor of $0.02$ for \cifar).  Next we examine the impact of
different distribution shifts on
backdoor survivability. Our goal is not to compare different
transformation types, but to examine how the volume of distribution
shifts affect backdoor survivability.
Figure~\ref{fig:datashift} plots, per transformation type, the backdoor $\gamma$-survivability
for one-shot poisoning on $D_0$, when varying the distribution shift step ($p$).
While $p$ is transformation-specific,  larger $p$
always means
heavier data distribution shifts over time.

\begin{figure}[t]
  \centering
  \vspace{-0.1in}
  \includegraphics[width=\linewidth]{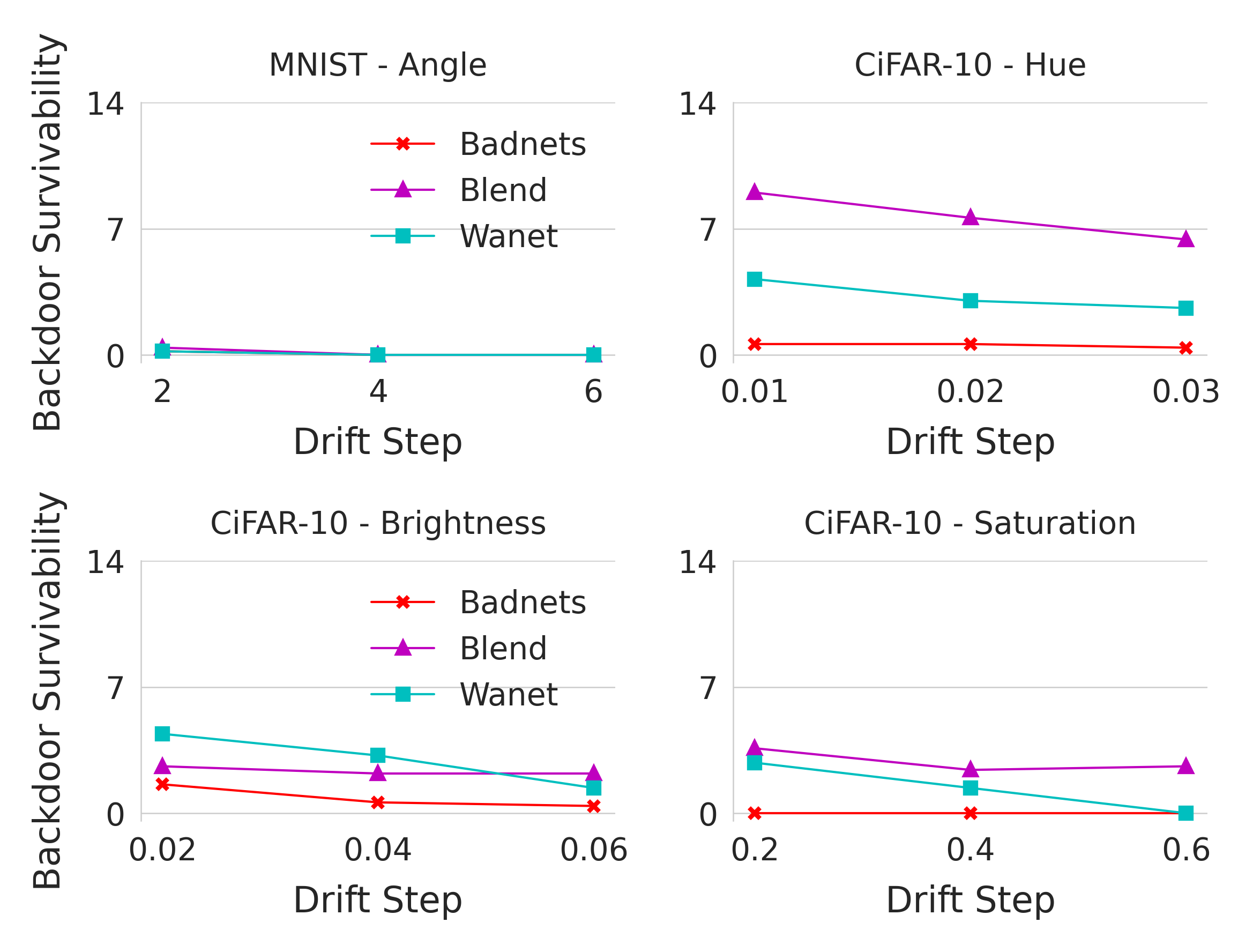}
  \vspace{-0.2in}
  \caption{Average backdoor survivability for one shot poisoning on $D_0$
  with different data distribution drift types and steps for the 3 attack methods on \mnist{} and \cifar{}.}
  \vspace{-0.2in}
  \label{fig:datashift}
\end{figure}

\para{Larger data distribution shifts accelerate backdoor forgetting.}
Our key observation is that, when the production data
experiences heavier changes over time, the backdoors are more
vulnerable to model updates. This is because when model updates
``force'' the time-varying model to learn more different data features, they also
accelerate the process of backdoor forgetting.  This observation also aligns with recent
work on continual
learning on old/new (benign) tasks~\cite{yin2020optimization},  which shows that
the larger the distance between two task distributions, the faster
the model forgets the old task it has learned previously. Different
from~\cite{yin2020optimization}, we focus on backdoor attacks rather
than benign tasks.  Overall, our findings further demonstrate the
importance of studying backdoor survivability,  since data
distribution shift is a common phenomenon in practice.

\begin{figure}[t]
\centering
  \subfigure[One-shot poison in $D_0$]{
  \includegraphics[width=.95\linewidth]{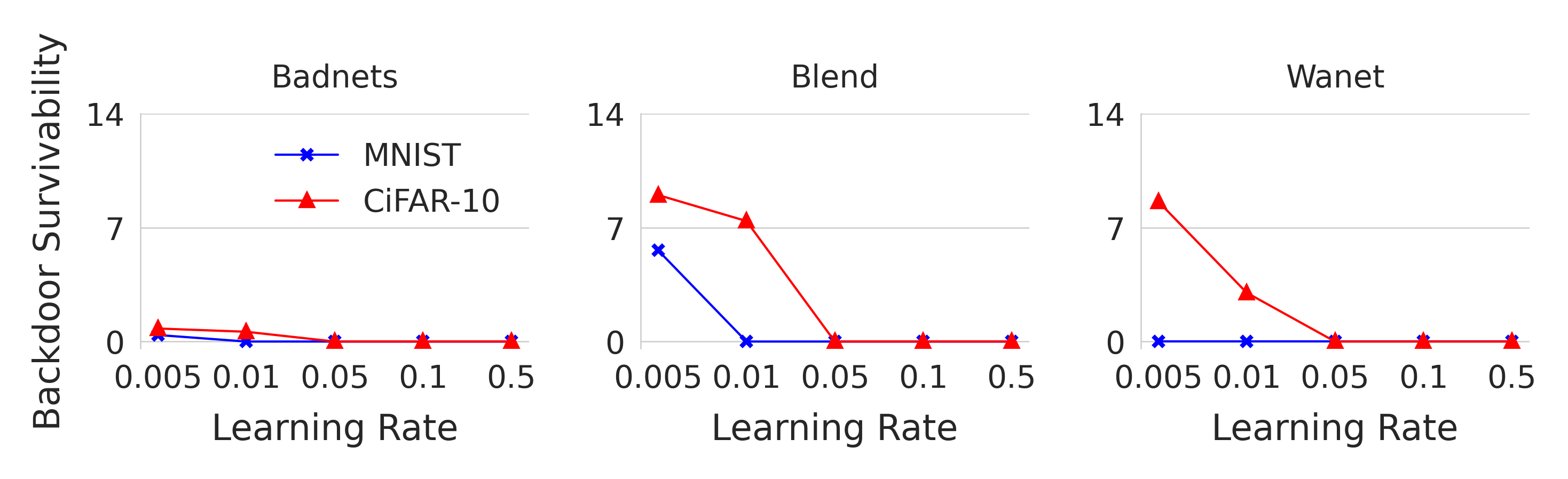}
  \label{fig:oneshot_D0_lr}
  \vspace{-0.1in}
  }
  \subfigure[One-shot poison in $D_1$]{
  \includegraphics[width=.95\linewidth]{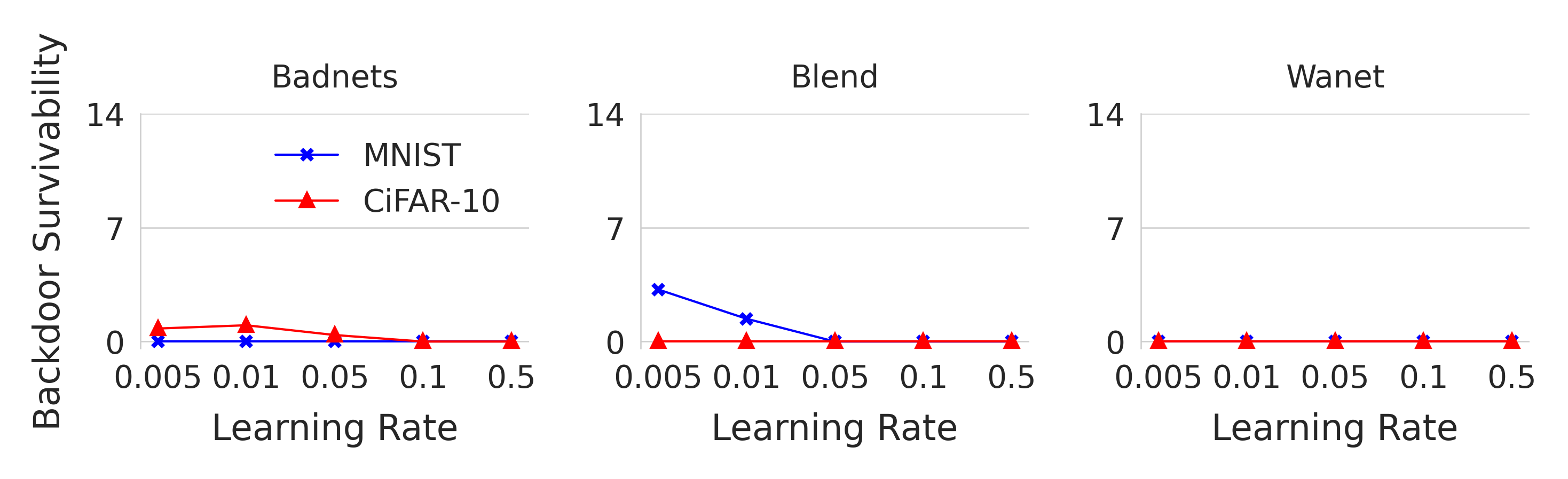}
  \label{fig:oneshot_D1_lr}
  \vspace{-0.1in}
  }\vspace{-0.1in}
\caption{Average backdoor survivability for one shot poisoning with different learning rates
during model updates for the 3 attack methods on \mnist{} and \cifar{}.}
\vspace{-0.1in}
\label{fig:lr}
\end{figure}

\begin{figure}[t]
  \centering
  \includegraphics[width=0.6\linewidth]{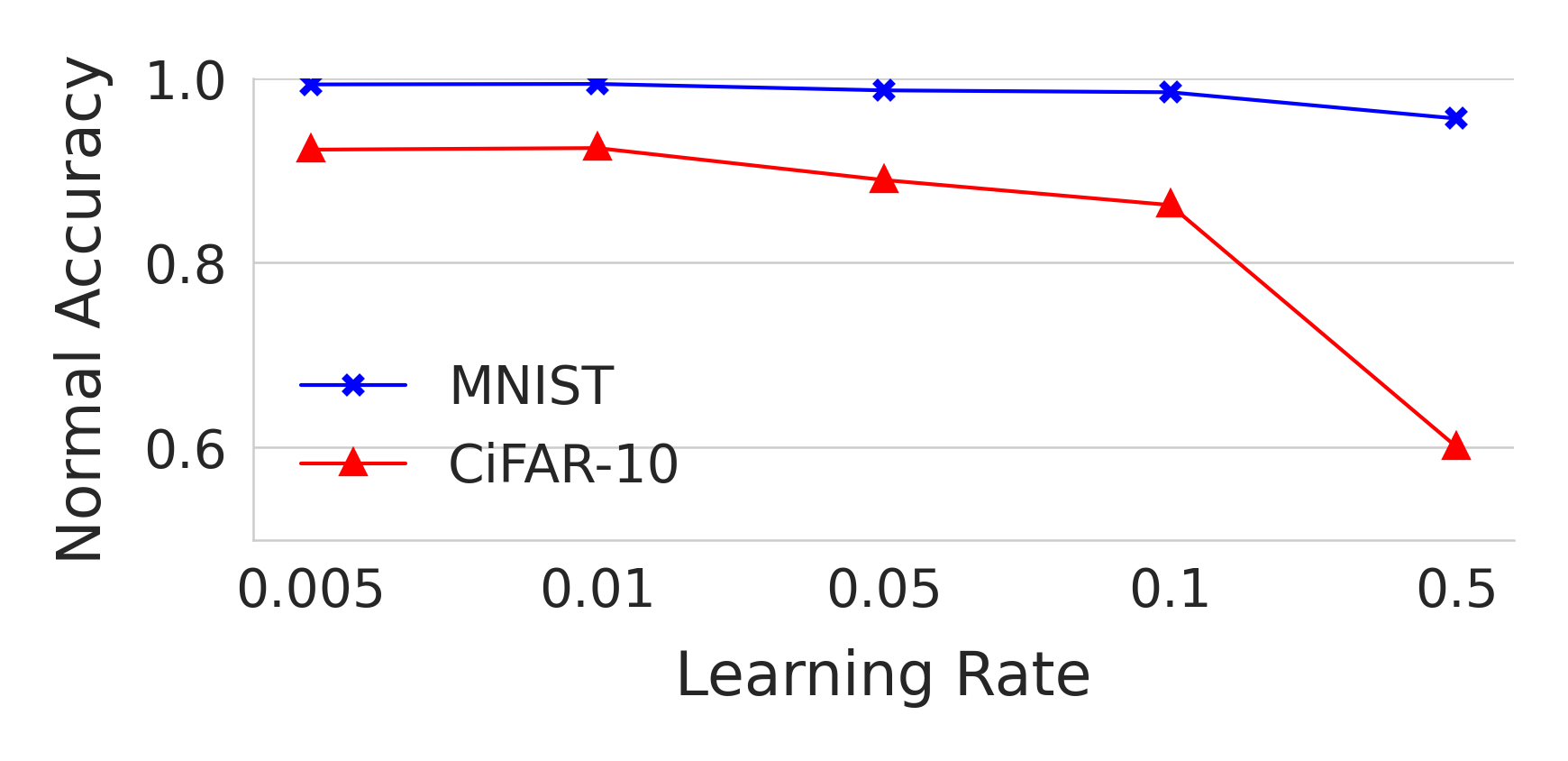}
  \vspace{-0.2in}
  \caption{Average normal accuracy after 15 model updates with different fine-tuning learning rates.}
  \vspace{-0.1in}
  \label{fig:acc_lr}
\end{figure}

%% file: training.tex
\vspace{-0.05in}
\section{Smart Training Strategies to Reduce Backdoor Survivability}
\label{sec:training}
Our analytical study implies that more training steps and reasonably larger learning rates 
can accelerate the backdoor forgetting process in Theorem~\ref{thm:backdoorremoval} and Corollary~\ref{coro:lr}.
We now empirically verify them on \mnist{} and \cifar{} datasets.
We start with increasing the training epochs and learning rates during the model
fine-tuning for each model update event and confirm that the backdoor survivability
degrades correspondingly. However, simply increasing learning rates for stochastic gradient descent (SGD)
will reduce the normal accuracy. As suggested in Corollary~\ref{coro:lr},
the learning rate cannot be too large as that will cause converge to slow down drastically. Thus,
we consider a smarter learning rate scheduler: Slanted Triangular Learning Rates
(STLRs)~\cite{howard2018universal}, which first increases the learning rate to
a very large value and then gradually decreases the learning rate for convergence.
Our results show that by using STLR the backdoor survivability significantly drops
while the normal accuracy stays as high as the initial model training.



 \para{Increasing training epochs and learning rate.} We empirically find
 that both training time (epochs) and learning rate can be better configured
 to reduce backdoor survivability without harming normal accuracy.  
We fine that increasing number of training epochs is inefficient, unstable 
 and costly in terms of decreasing backdoor survivibility (detailed results in Appendix~\ref{subsec:app_te}). 
 In the meanwhile, increasing learning rate has much more impact on
 backdoor survivability. As shown in  Figure~\ref{fig:lr}, the backdoor survivability decreases all to 0
 when increasing the learning rate from 0.005 to 0.5,.
 However, larger learning rates may prevent models from convergence. As shown in
 Figure~\ref{fig:acc_lr}, we can see that the model normal accuracy after 15 model
 updates drops significantly when we increase the learning rate ($99\%$ to $91\%$ for \mnist{}
 and $92\%$ to $60\%$ for \cifar{}).

\begin{figure}[t]
\centering
  \subfigure[One-shot poison in $D_0$]{
  \vspace{-0.1in}
  \includegraphics[width=\linewidth]{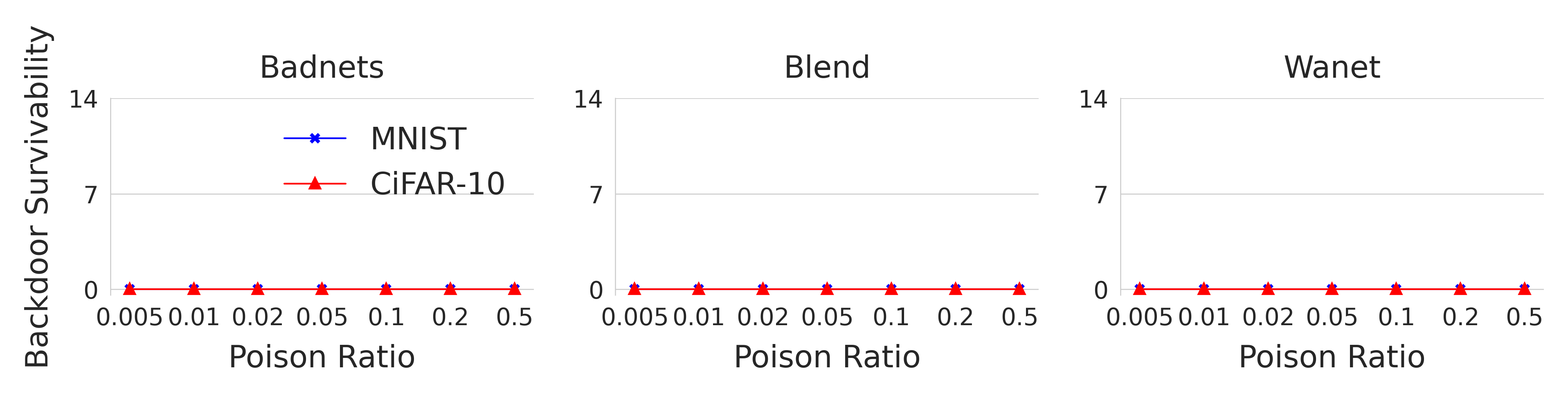}
  \vspace{-0.1in}
  \label{fig:oneshot_D0_pr_stlr}
  }
  \vspace{-0.1in}
  \subfigure[One-shot poison in $D_1$]{
  \vspace{-0.05in}
  \includegraphics[width=\linewidth]{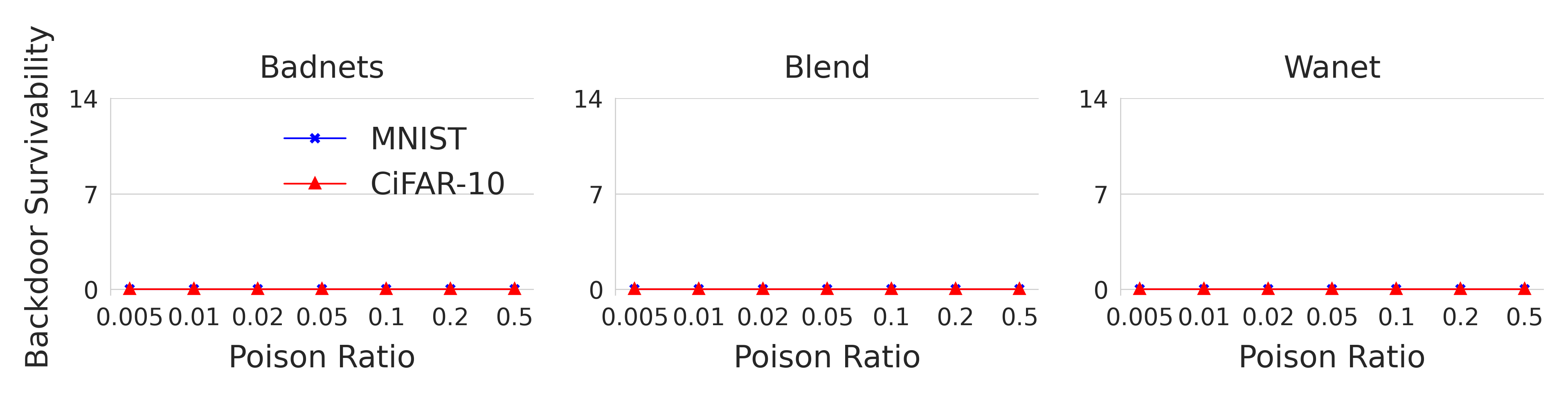}
  \label{fig:oneshot_D1_pr_stlr}
  }
  \vspace{-0.1in}
\caption{Average backdoor survivability for one shot poisoning
using different poison ratios with STLR with max lr = 0.5.}
\label{fig:oneshot_pr_stlr}
\vspace{-0.15in}
\end{figure}

\begin{figure}[t]
\centering
  \subfigure[One-shot poison in $D_0$]{
  \vspace{-0.1in}
  \includegraphics[width=.95\linewidth]{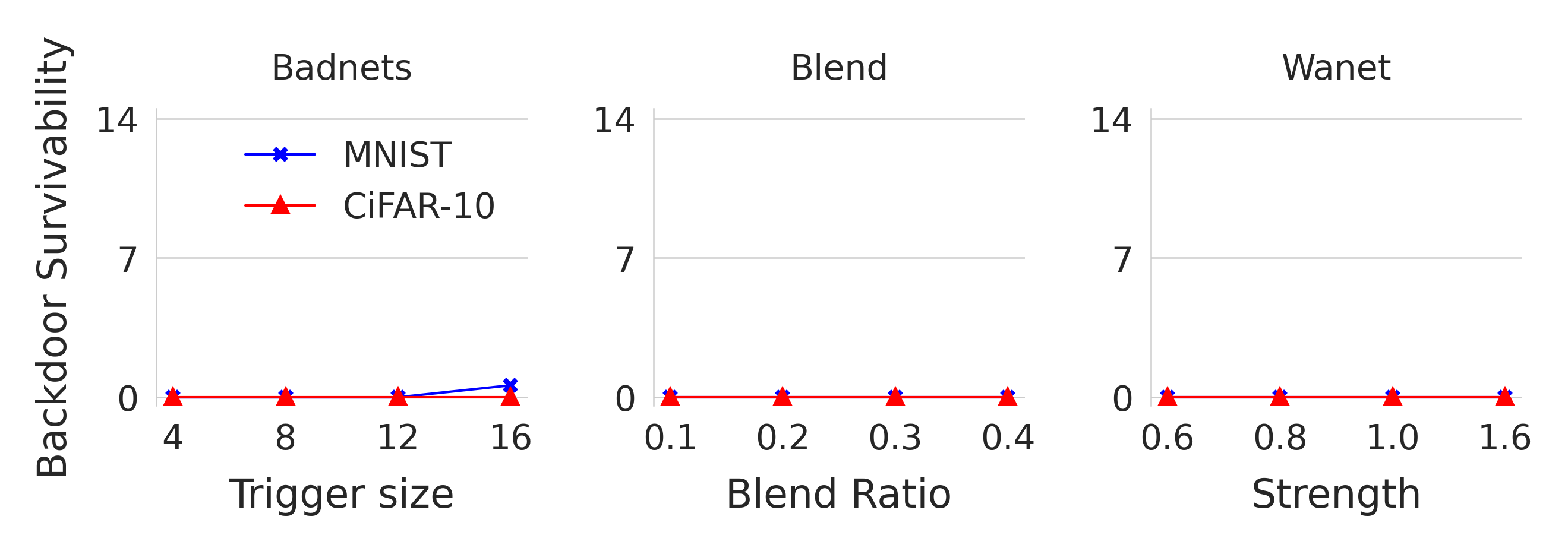}
  \vspace{-0.1in}
  \label{fig:oneshot_D0_te_stlr}
  \vspace{-0.1in}
  }
  \subfigure[One-shot poison in $D_1$]{
  \vspace{-0.1in}
  \includegraphics[width=.95\linewidth]{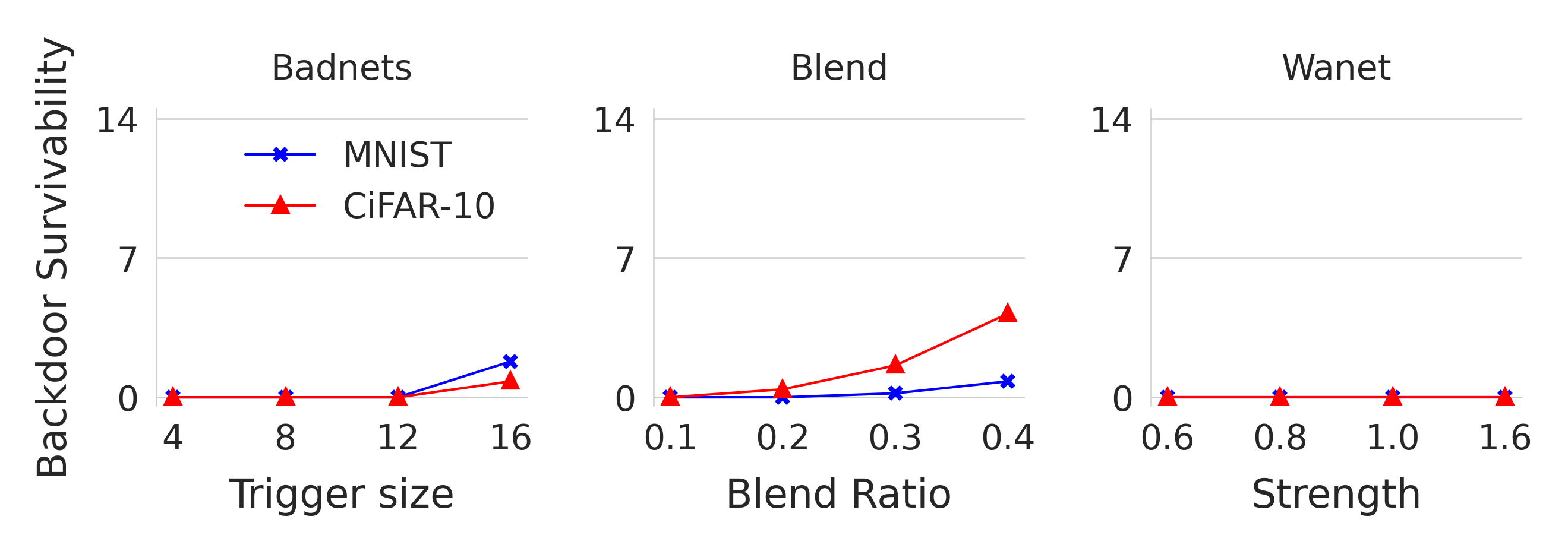}
  \label{fig:oneshot_D1_te_stlr}
  \vspace{-0.1in}
  }\vspace{-0.1in}
\caption{Average backdoor survivability for one shot poisoning using different triggers
 with STLR with max lr = 0.5.}
\vspace{-0.15in}
\label{fig:te_stlr}
\end{figure}

\vspace{-0.05in}
\para{Smarter learning rate scheduler (STLR).}  Inspired by Corollary~\ref{coro:lr},
we consider a smarter learning rate scheduler: Slanted Triangular Learning Rates
(STLRs)~\cite{howard2018universal}.  It first linearly increases the learning
rate from 0 to the max value (e.g., up to 0.5) and then linearly decays it
back to 0.  Interestingly, STLR is highly effective in terms of reducing
backdoor survivability.


Figure~\ref{fig:oneshot_pr_stlr} lists the backdoor
survivability for one-shot poisoning attacks (on $D_1$ and
$D_0$), when updating models using STLR.   The baseline version (e.g.,
model updates using SGD) is in
Figure~\ref{fig:oneshot_pr}. Comparing the two figures, we see that
training using STLR naturally removes the backdoor, i.e., an
injected backdoor cannot survive past a single model update.
Figure~\ref{fig:te_stlr} shows that STLR significantly reduces the
backdoor survivability for all different types of triggers (compared to Figure~\ref{fig:te}).
At the same time, STLR maintains the normal accuracy at a high level after 15 model updates ($99\%$ for
\mnist{} and $91\%$ for \cifar{}).

\para{Summary.}  We find that model owners can significantly reduce backdoor
survivability by leveraging adaptive learning rate schedulers like STLR, which 
is much more efficient compared to increasing training epochs. 
In most cases, all three backdoor survivability
metrics drop to 0, meaning the backdoor is immediately removed by a single
fine-tuning update. This additional benefit only adds to the existing value
of these techniques as ways to enhance model accuracy.

%% file: conclusion.tex
\vspace{-0.1in}
\section{Discussion}
\label{sec:conclusion}

In this paper, we take a first look at backdoor attacks in the real-world
context of \tvmodels{}. Guided by our theoretical analysis,
we introduce the first empirical metric on a backdoor's survivability on \tvmodels, study a range of factors
that may impact backdoor survivability, and propose an intelligent training strategy using
STLR to significantly reduce backdoor survivability.

Looking forward, our study also presents potential directions for future research
on the topic of backdoor attacks on \tvmodels. We discuss some below. 

\para{Novel defenses against backdoor attacks on \tvmodels.} Existing work has produced numerous defenses against backdoors in static models, e.g.,~\cite{gao2019strip, finepruning, wang2019neural, guo2020towards, liu2019abs,
qiao2019defending}. A natural question arises: ``{\em are existing backdoor defenses
effective on \tvmodels?}''  We believe the answer is no, and we need new defenses specifically designed for \tvmodels.



We believe that \tvmodels{} break common assumptions used in existing backdoor defenses. 
For example, model inspection defenses such as \nc{}~\cite{wang2019neural}
and \tabor{}~\cite{guo2020towards} try to identify the backdoor trigger as the smallest perturbation
that makes the model misclassify all perturbed inputs to a target label. If there is a backdoor in the model, the computed perturbation for the target label should be much smaller
than other benign labels, i.e., an anomaly. These defenses implicitly assume that an existing backdoor has a very high (e.g. 100\%) attack success rate. 
However, we know that fine-tuning can reduce attack success rate to unpredictable levels over time, 
making it harder for defenses to reliably detect the anomalous misclassification behavior.
We can adapt by lowering the attack success rate threshold to 75\%, 50\% and 25\%, but this approach does not work. The average
anomaly index for \tabor{} on clean models increases to 3.31 when setting the attack success rate threshold to 25\%
and average anomaly index for \nc{} on clean models increases to 2.27 when setting the attack success rate threshold to 50\%
(any model with an anomaly index over 2 is considered a backdoored model).
Not only is setting new threshold values challenging, but doing so also lowers the detection performance for static backdoored models.
More detailed results can be found in Appendix~\ref{subsec:defense}.

Another example is {\fp}~\cite{finepruning}, a model sanitization defense. It prunes a portion of neurons
from the model, and then fine-tunes the model with a small set of pure clean data. However, this approach
is problematic for time-varying models, as it would require continuous pruning of neurons, and continued pruning will lower model performance over time~\cite{he2016deep, tan2019efficientnet}.

\para{Limitations and future work.}
As the first study on backdoor survivability, our work faces several
limitations. First, our study focused on \tvmodels{} updated via fine-tuning (transfer learning). 
More research is needed to understand survivability of backdoors in \tvmodel{} 
updated via other mechanisms.
Second, our theoretical results consider fine-tuning with data drawn from the original distribution, 
and we leave the analysis of distribution drift on backdoor forgetting for future work.
Third, we focused on two commonly image datasets to demonstrate the property of ``backdoor
forgetting.''  We need more empirical and analytical work to characterize the relationship between
model updates and backdoor forgetting on other domains.  Finally, we trigger model
updates using controlled data dynamics via image
transformation. Experiments using a broader and more realistic category of data
dynamics may provide more insights on how to leverage natural data
variations (and thus model variations) to resist backdoor attacks.



%% file: appendix.tex
\newpage
\appendix
\onecolumn
We first present the proofs for our theoretical analysis in Appendix~\ref{sec:app_theory}.
Then, we give more details about our experimental setups in Appendix~\ref{sec:app_setup}, followed by more empirical
results on backdoor survivability in Appendix~\ref{sec:app_expr_results}.

\section{Proofs}
\label{sec:app_theory}

\poisonslowdown*

\begin{proof}[Proof of Theorem~\ref{lemma:poison_slowdown}]
Since $L_{D_0}$ and $L_p$ are both strongly convex by assumption, $L_{\text{mix}}= \alpha L_p + (1- \alpha) L_{D_0}$ is also strongly convex. From \citet{boyd2004convex}, the strong convexity parameter of $L_{\text{mix}}$ is $\sigma_{\text{mix}} = \alpha \sigma_p + (1-\alpha) \sigma_b$.

Then, from Eq. (4) of \citet{kuwaranancharoen2018location} on necessary conditions for the minimizer of the sum of two strongly convex functions, we have, for all $\btheta \in \Theta$,

\begin{align}
\begin{split}
    \left( \nabla L_{\text{mix}}(\btheta) - \nabla L_{\text{mix}}(\btheta_{\text{mix}}) \right)^\intercal (\btheta - \btheta_{\text{mix}}) &\geq \sigma_{\text{mix}} \| \btheta - \btheta_{\text{mix}} \|^2 \\
    \Rightarrow \left( \nabla L_{\text{mix}}(\btheta^*_0) - \nabla L_{\text{mix}}(\btheta_{\text{mix}}) \right)^\intercal \frac{(\btheta^*_0 - \btheta_{\text{mix}})}{\| \btheta^*_0 - \btheta_{\text{mix}}) \|}  &\geq \sigma_{\text{mix}} \| \btheta^*_0 - \btheta_{\text{mix}} \|
\end{split}
\end{align}

Since $\btheta_{\text{mix}}$ is the minimizer of $L_\text{mix}$ and $\btheta^*_0$ is the minimizer of $L^*_0$, their gradients are zero at those points. Thus, we get,

\begin{align}
    \frac{\alpha \left( \nabla L_p(\btheta_0^*)^\intercal \frac{(\btheta^*_0 - \btheta_{\text{mix}})}{\| \btheta^*_0 - \btheta_{\text{mix}}) \|}  \right) }{\alpha \sigma_p + (1-\alpha) \sigma_b} \geq \| \btheta^*_0 - \btheta_{\text{mix}} \|.
\label{eq:alpha_bound}
\end{align}
Taking the norm on both sides, we get the statement of the lemma.

\end{proof}

We can immediately notice that the L.H.S of Eq.~\ref{eq:alpha_bound} is an increasing function of $\alpha$.




\backdoorremoval*

\begin{proof}[Proof of Theorem~\ref{thm:backdoorremoval}]
From Lemma 1 from \citet{rakhlin2011making}, we know 
\begin{align}
    \mathbb{E} \left [ \| \btheta_t - \btheta_0^* \| \right ] \leq \left ( 1-\frac{2}{t} \right ) \mathbb{E} \left [ \| \btheta_{t-1} - \btheta_0^* \|\right ] + \frac{\gamma_b^2}{\sigma_b^2 t^2} .
\end{align}
By induction, we get 
\begin{align}
    \mathbb{E} \left [ \| \btheta_T - \btheta_0^* \| \right ] \leq \frac{\max \{ \| \btheta_{\text{init}} - \btheta_0^* \| , \frac{\gamma_b^2}{\sigma_b^2} \}}{T}.
\end{align}

With $\btheta_{\text{init}}=\btheta_{\text{mix}}$, assuming that $\| \btheta_{\text{mix}} - \btheta_0^* \| > \frac{\gamma_b^2}{\sigma_b^2}$, and using the bound on $\| \btheta_{\text{mix}} - \btheta_0^* \|$ from Lemma~\ref{lemma:poison_slowdown}, we get

\begin{align}
    \mathbb{E} \left [ \| \btheta_T - \btheta_0^* \| \right ] \leq \frac{    \frac{\alpha \left( \nabla L_p(\btheta_0^*)^\intercal \frac{(\btheta^*_0 - \btheta_{\text{mix}})}{\| \btheta^*_0 - \btheta_{\text{mix}}) \|}  \right) }{\alpha \sigma_p + (1-\alpha) \sigma_b}}{T}
    \label{eq:final_bound}
\end{align}

The theorem is obtained by combining Eq.\ref{eq:final_bound} with the assumption on the $\gamma_b$-smooth nature of $L_{D_0}$.
\end{proof}

\begin{table}[]
  \centering
  \resizebox{0.48\textwidth}{!}{
\begin{tabular}{c|ccc}
\hline
Attack                                   & \begin{tabular}[c]{@{}c@{}}Training\\ Data\end{tabular} & \begin{tabular}[c]{@{}c@{}}Model\\Access\end{tabular} & \begin{tabular}[c]{@{}c@{}}Training\\ Process\end{tabular} \\ \hline
\cite{gu2017badnets}    & \checkmark                                      &                                                        &                                                                   \\
\cite{chen2017targeted} & \checkmark                                      &                                                        &                                                                   \\
\cite{liu2018trojaning}  &   \checkmark                                      &       \checkmark                                                 &      \checkmark                                                             \\
\cite{yao2019latent}  &    \checkmark                                     &       \checkmark                                                 &      \checkmark                                                      \\
\cite{nguyen2020input}  & \checkmark                                      &  \checkmark                                                      &        \checkmark                                                           \\
\cite{tang2020embarrassingly}  &                                      &    \checkmark                                                     &                                                                   \\
\cite{liu2020reflection}  &       \checkmark                               &    \checkmark                                                     &                                                                  \\
\cite{pang2020tale}  &          \checkmark                            &    \checkmark                                                     &        \checkmark                                                          \\
\cite{nguyen2021wanet}  & \checkmark                                      &                                                        &                                                                   \\ \hline
\end{tabular}}
\caption{Summary of representative backdoor attacks and their threat models.}
\label{tab:attacks}
\end{table}

\begin{table}[t]
  \centering
  \begin{tabular}{l|l}
    \hline
  Transformation Types  &  Function Call \\ \hline
  Angle    & {\tt rotate($X$, $i \times p$)} \\
  Brightness    & {\tt adjust\_brightness($X$, $1+i \times p$)} \\
  Hue    & {\tt adjust\_hue($X$, $i \times p$)} \\
  Saturation    & {\tt adjust\_saturation($X$, $1+i \times p$)} \\\hline
\end{tabular}
  \caption{Function calls for generating training datasets $D_i$ for different data distribution drifts with given shift steps  $p$. $X$ is the original training data with no transformations.}
  \label{tab:shift_setup}
\end{table}

\begin{figure}[h]
\centering
  \subfigure[\mnist{}]{
  \includegraphics[width=.2\linewidth]{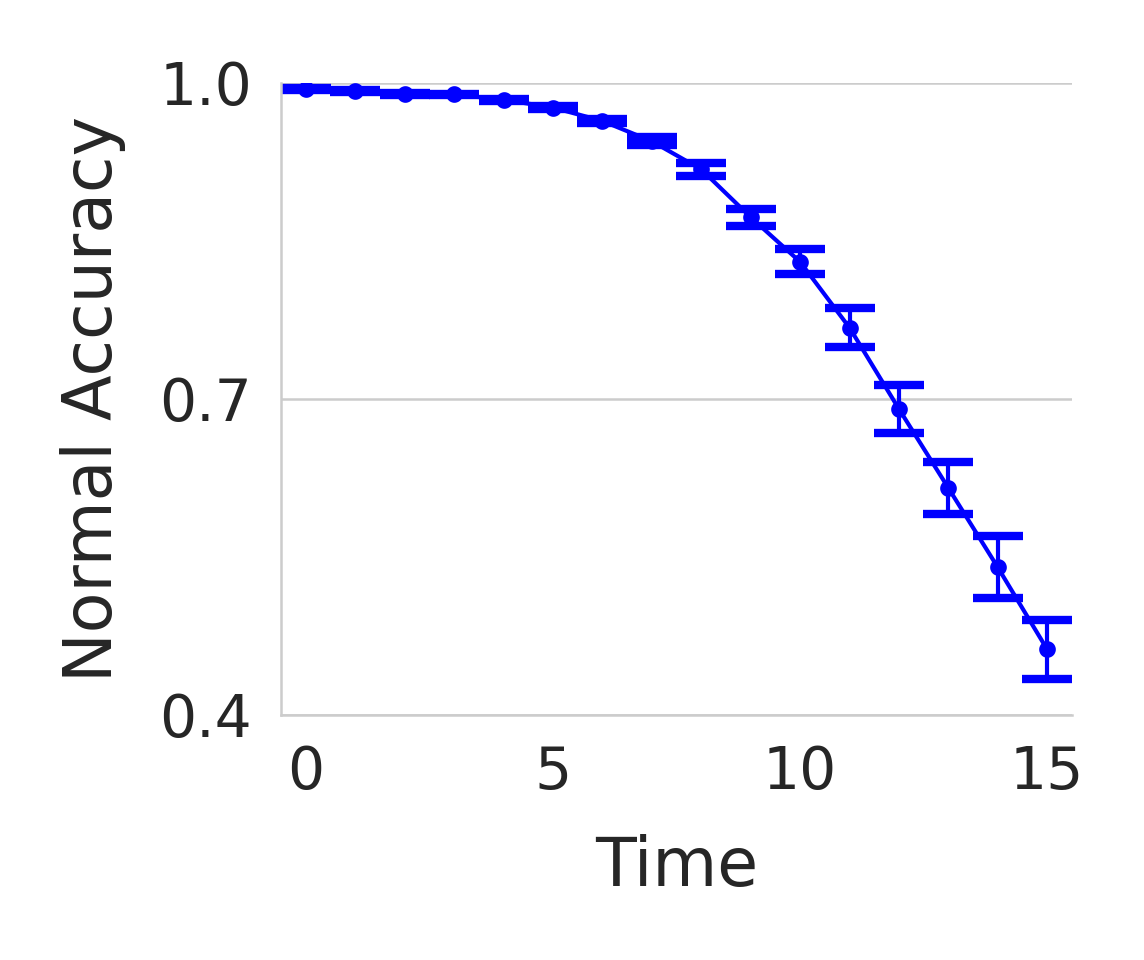}
  \label{fig:mnist_acc_drop}
  }
  \subfigure[\cifar{}]{
  \includegraphics[width=.2\linewidth]{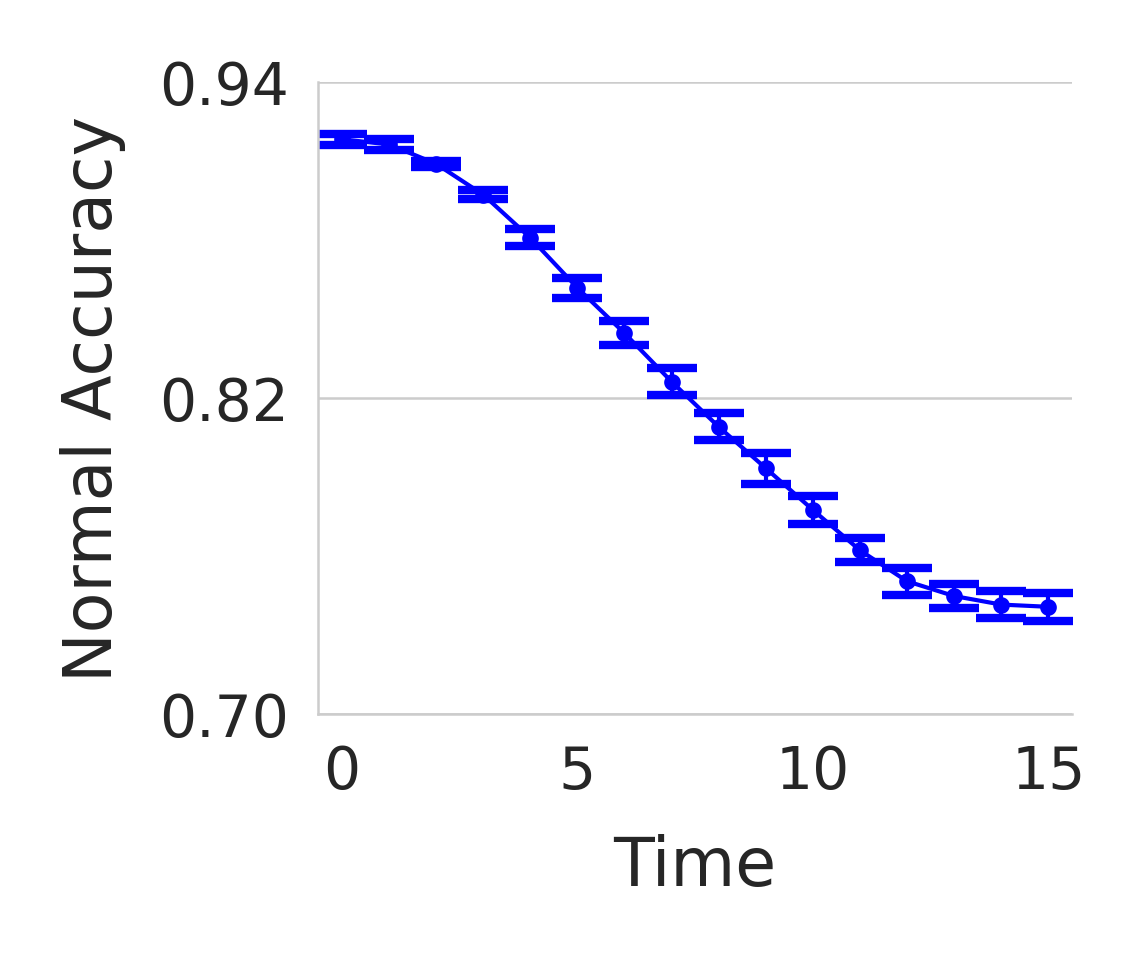}
  \label{fig:cifar10_acc_drop}
  }\vspace{-0.05in}
\caption{Normal accuracy for a static model when inference data distribution drifts.
For \cifar{} we change the hue by a factor of $0.02$ per drift step and for \mnist{} we change
the angle for $4^\circ$ per drift step. We report the mean with std for 5 models on each
dataset.}
\label{fig:acc_drop}
\end{figure}

\begin{figure}[h]
\centering
  \subfigure[One-shot poison in $D_0$]{
  \includegraphics[width=.45\linewidth]{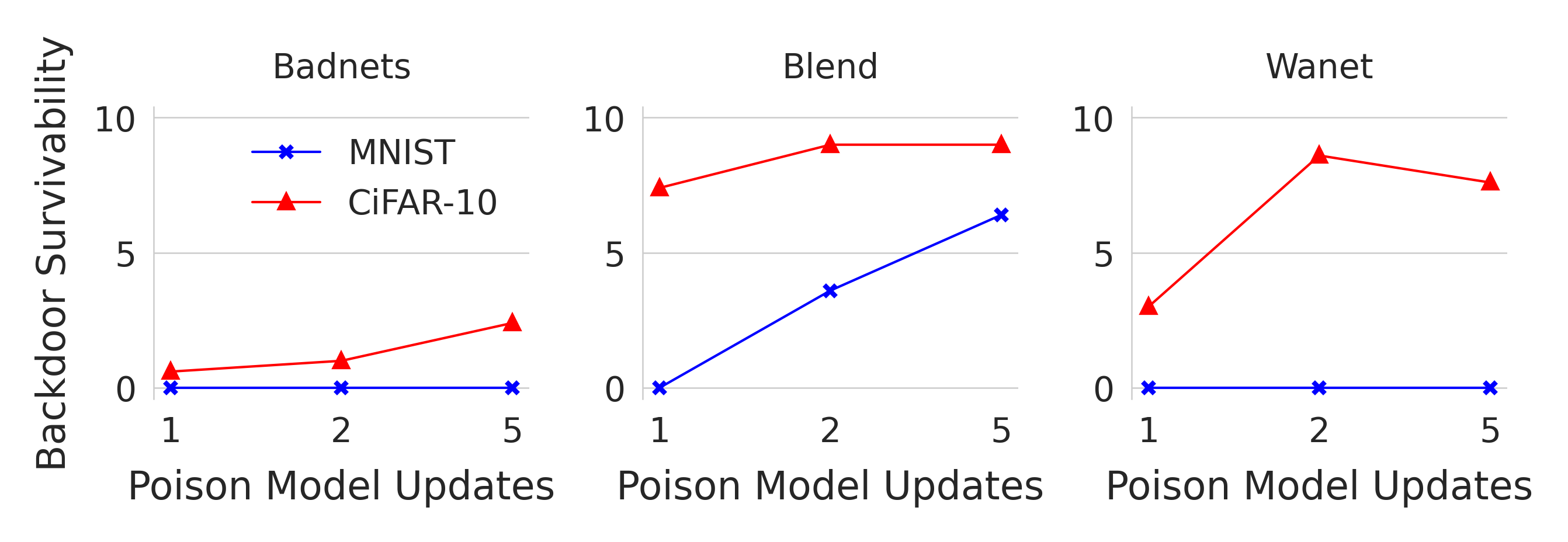}
  \label{fig:oneshot_D0_pu}
  \vspace{-0.1in}
  }
  \subfigure[One-shot poison in $D_1$]{
  \includegraphics[width=.45\linewidth]{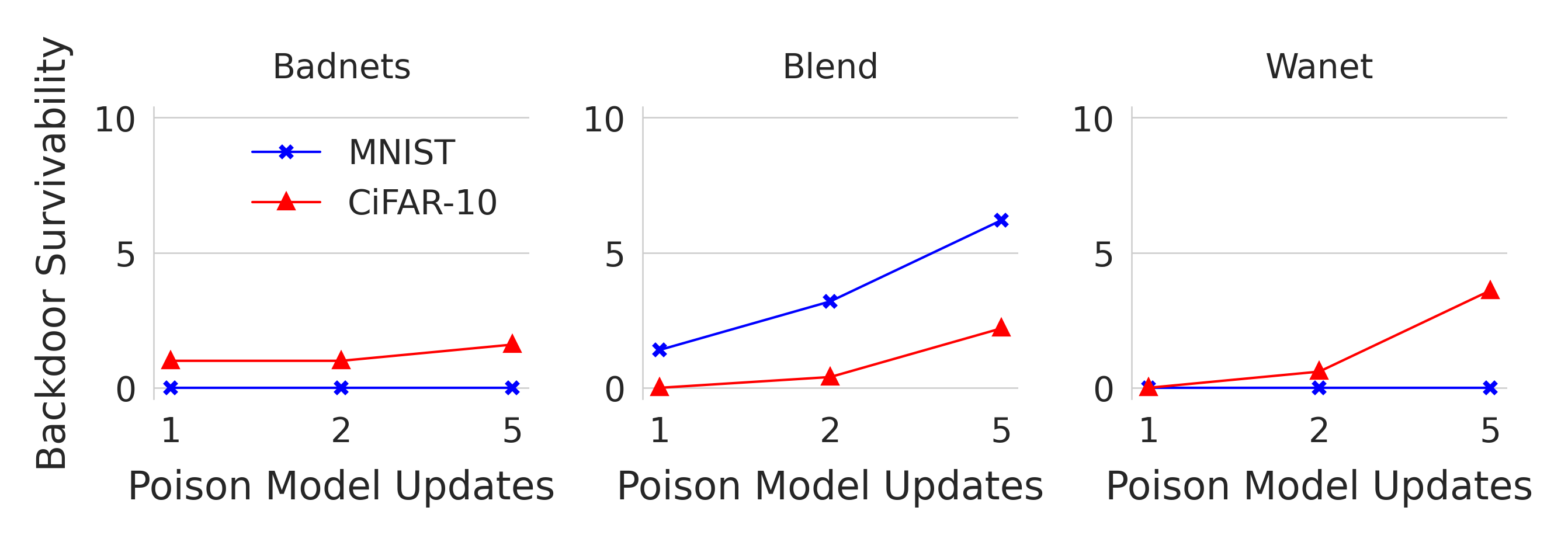}
  \label{fig:oneshot_D1_pu}
  \vspace{-0.1in}
  }\vspace{-0.05in}
\caption{Backdoor survivability for persistent poisoning with different poison
model updates for the 3 attack methods on \mnist{} and \cifar{}.
The results are averaged on 5 instances.}
\vspace{-0.1in}
\label{fig:pu}
\end{figure}

\begin{figure}[h]
\centering
  \subfigure[One-shot poison in $D_0$]{
  \includegraphics[width=.45\linewidth]{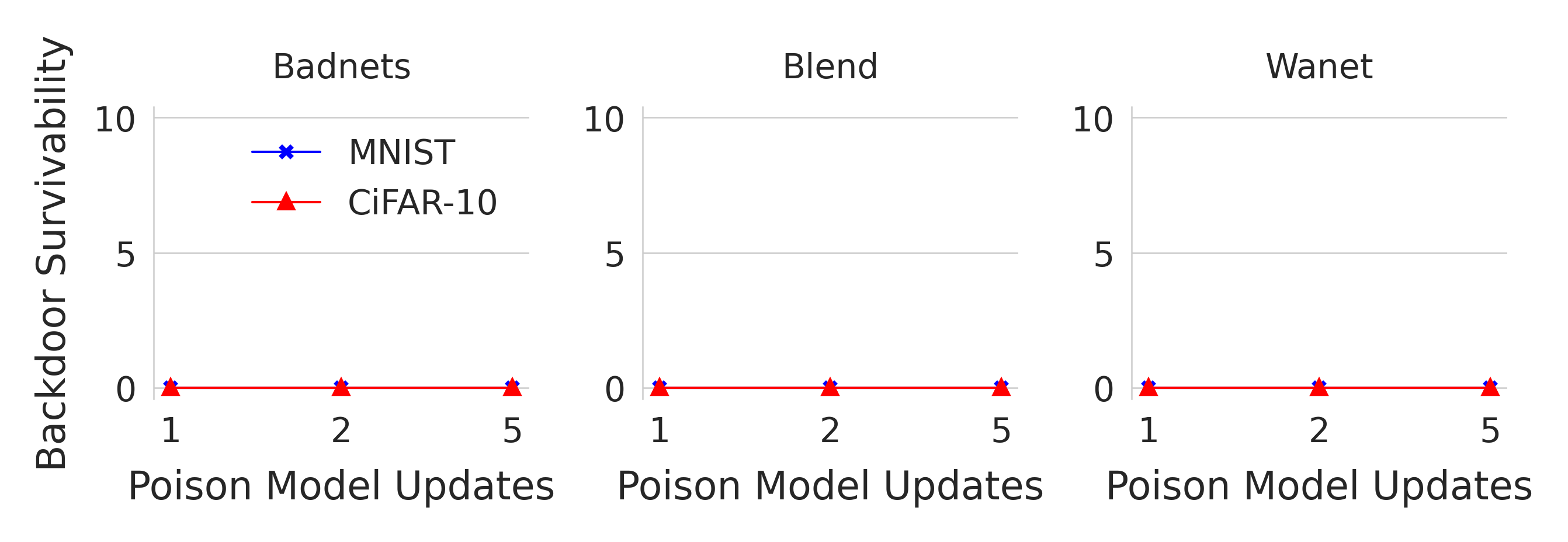}
  \label{fig:oneshot_D0_pu_stlr}
  \vspace{-0.1in}
  }
  \subfigure[One-shot poison in $D_1$]{
  \includegraphics[width=.45\linewidth]{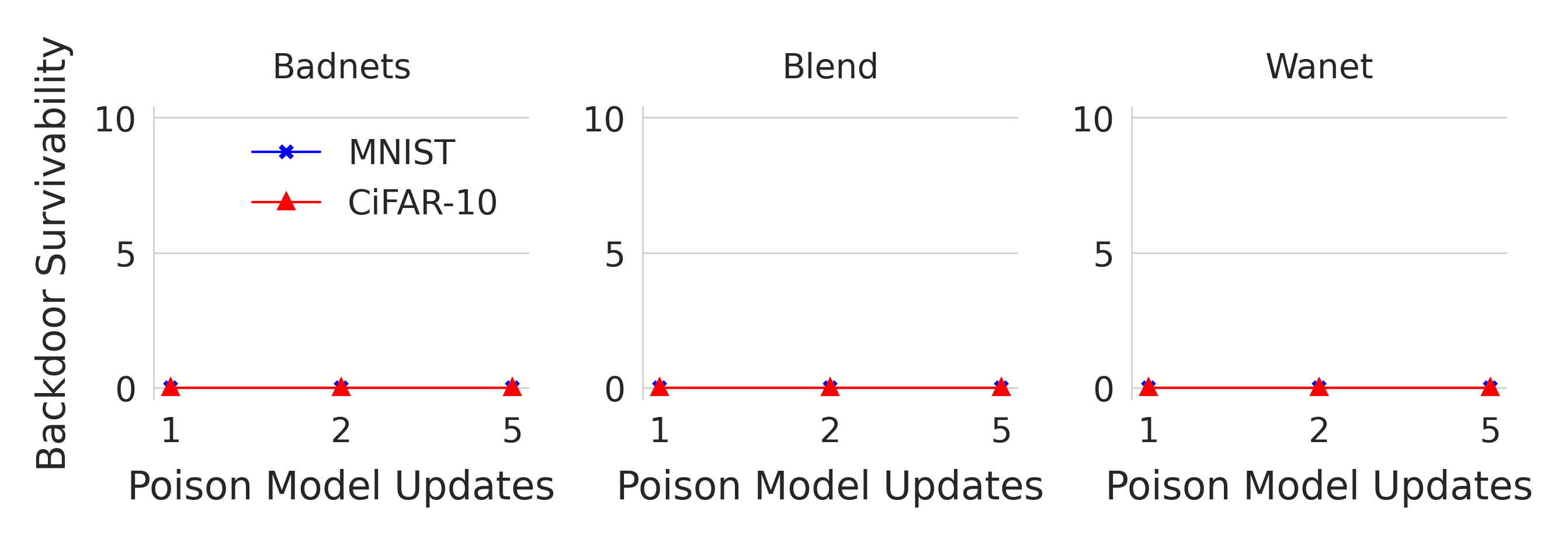}
  \label{fig:oneshot_D1_pu_stlr}
  \vspace{-0.1in}
  }\vspace{-0.05in}
\caption{Backdoor survivability for persistent poisoning with different poison
model updates with STLR.}
\vspace{-0.1in}
\label{fig:pu_stlr}
\end{figure}

\begin{figure}[h]
\centering
  \subfigure[One-shot poison in $D_0$]{
  \includegraphics[width=.45\linewidth]{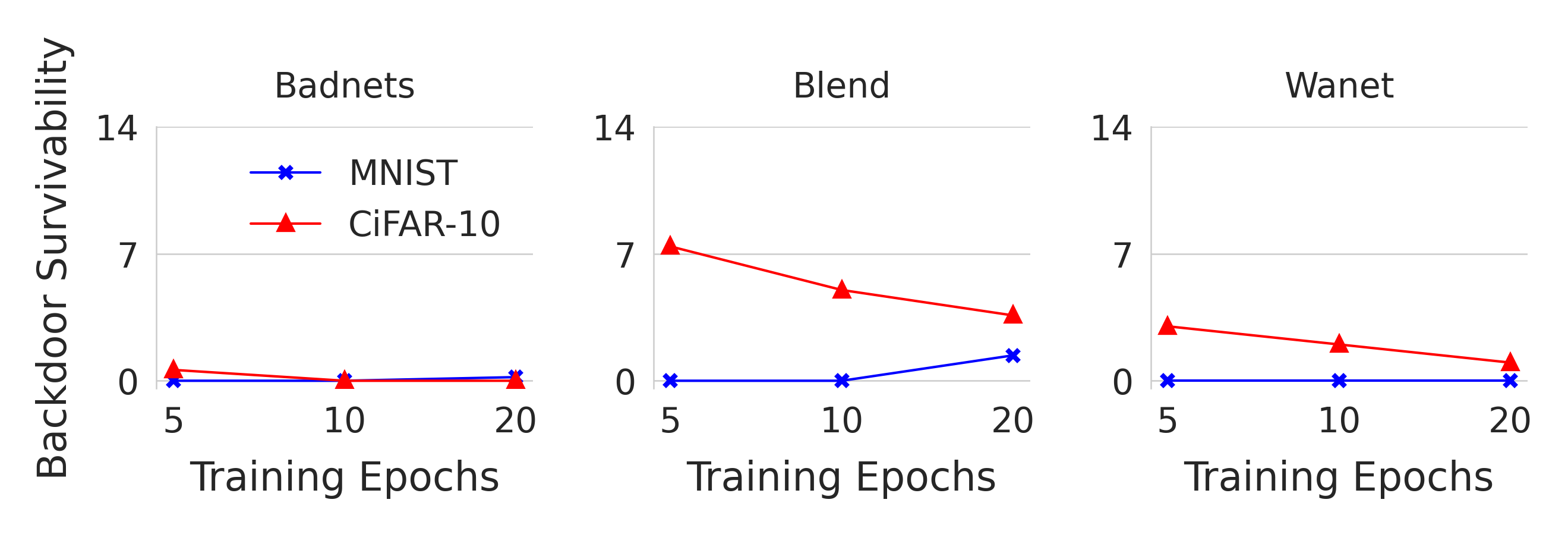}
  \label{fig:oneshot_D0_ne}
  \vspace{-0.1in}
  }
  \subfigure[One-shot poison in $D_1$]{
  \includegraphics[width=.45\linewidth]{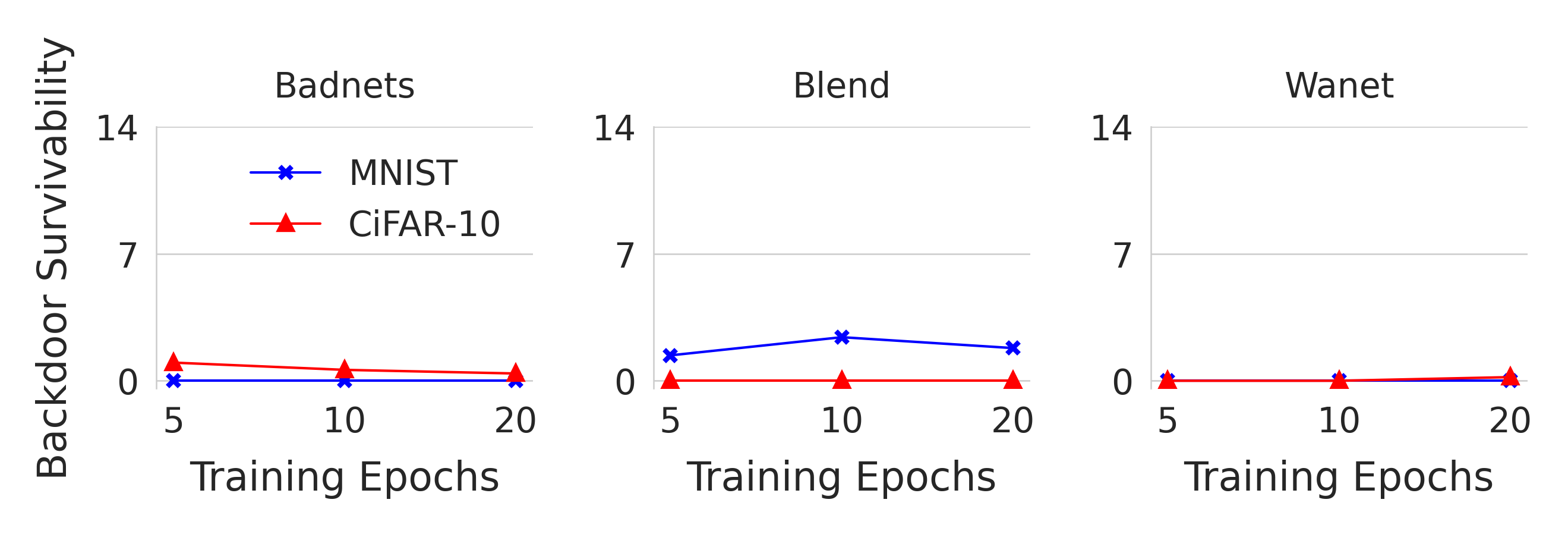}
  \label{fig:oneshot_D1_ne}
  \vspace{-0.1in}
  }\vspace{-0.05in}
\caption{Average backdoor survivability for persistent poisoning with different training
epochs during model updates for the 3 attack methods on \mnist{} and \cifar{}.}
\vspace{-0.2in}
\label{fig:ne}
\end{figure}

\begin{table}[h]
  \centering
\begin{minipage}[h]{0.45\textwidth}
\centering
  \resizebox{0.7\textwidth}{!}{
\begin{tabular}{c|c|c}
\hline
Models  & \nc{}   & \tabor{} \\ \hline
Clean   & 1.52 & 1.84  \\
Badnets & 3.13 & 2.12  \\
Blend   & 3.35 & 3.11  \\
Wanet   & 2.44 & 2.01  \\ \hline
\end{tabular}}
\caption{Average anomaly index for \nc{} and \tabor{} on clean models and backdoored models ($F_0$).
The attack success rate threshold is set to 99\%. For each type of models, we average the anomaly index over 5 models.}
\label{tab:defense_f0_99}
\end{minipage}
\hfill
\begin{minipage}[h]{0.45\textwidth}
\centering
  \resizebox{0.7\textwidth}{!}{
\begin{tabular}{c|c|c}
\hline
Models  & \nc{}   & \tabor{} \\ \hline
Badnets & {\bf 1.00} & {\bf 1.54}  \\
Blend   & {\bf 1.52} & 2.12  \\
Wanet   & {\bf 0.87} & {\bf 1.45}  \\ \hline
\end{tabular}}
\caption{Average anomaly index for \nc{} and \tabor{} on fine-tuned backdoored models after poison stops ($F_1$, the
backdoors are injected in $F_0$). The attack success rate threshold is set to 99\%. {\bf Bold numbers} indicate failures to
detect the backdoors.}
\label{tab:defense_f1_99}
\end{minipage}
\vfill
\vspace{5pt}
\begin{minipage}[h]{0.9\textwidth}
  \centering
  \resizebox{0.9\textwidth}{!}{
    \begin{tabular}{c|cc|cc|cc}
    \hline
    \multirow{3}{*}{\diagbox{Models}{Attack Success \\ Rate  Threshold}} & \multicolumn{2}{c|}{75\%}        & \multicolumn{2}{c|}{50\%}        & \multicolumn{2}{c}{25\%}         \\
                                                                        & \multicolumn{2}{c|}{}        & \multicolumn{2}{c|}{}        & \multicolumn{2}{c}{}         \\ \cline{2-7}
                                                                               & \multicolumn{1}{c|}{\nc}   & \tabor & \multicolumn{1}{c|}{\nc}   & \tabor & \multicolumn{1}{c|}{\nc}   & \tabor \\ \hline
    Clean                                                                     & \multicolumn{1}{c|}{1.26} & 1.90  & \multicolumn{1}{c|}{{\bf 2.27}} & 1.34  & \multicolumn{1}{c|}{1.53} & {\bf 3.31}  \\
    Badnets                                                                   & \multicolumn{1}{c|}{3.38} & 2.89  & \multicolumn{1}{c|}{3.23} & 2.06  & \multicolumn{1}{c|}{{\bf 1.70}} & 2.41  \\
    Blend                                                                     & \multicolumn{1}{c|}{2.29} & 2.46  & \multicolumn{1}{c|}{2.11} & 3.02  & \multicolumn{1}{c|}{{\bf 1.28}} & 2.01  \\
    Wanet                                                                     & \multicolumn{1}{c|}{{\bf 1.64}} & {\bf 1.86}  & \multicolumn{1}{c|}{{\bf 1.53}} & {\bf 1.52}  & \multicolumn{1}{c|}{{\bf 1.63}} & {\bf 1.51}  \\ \hline
  \end{tabular}}
\caption{Average anomaly index for \nc{} and \tabor{} on clean and backdoored models ($F_0$)
with different attack success rate threshold. {\bf Numbers in bold} indicate instances
where either the backdoors were not detected or false positive detections of backdoors were made on clean models.}
\label{tab:defense_f0}
\end{minipage}
\end{table}

\section{Details for Experimental Setup}
\label{sec:app_setup}
In this section, we provide further information about our experimental setup,
expanding upon the background from \S\ref{subsec: exp_setup} of the main paper.

\para{More details about distribution shifts.}
We select 4 types of image transformations, namely: i) changing angle, ii)
changing brightness, iii) changing hue, and iv) changing
saturation. These transformations also reflect practical
scenarios when the camera's color spectrum varies due to hardware aging or dust
accumulation or a deployed camera's view is shadowed by a new structure nearby
or the angle of the camera is rotated. We use these transformations to introduce fine-grained,
parameterized data distribution shifts over the sequence of changing
data distributions ($\{D_i\}_{i>0}$). When implementing the
transformations, we use PyTorch's built-in functions in {\tt
torchvision.transforms.functional} (see Table~\ref{tab:shift_setup}).

\para{More details on model training and updating.}
To train an initial model $F_0$, we choose ResNet-9  as
our default model architecture and train the model using $D_0$
with a batch-size of $512$. We train the model for 80 epochs for \cifar{} and 40 epochs for \mnist{} (note the 
initial training datasets for both tasks are half of the original training datasets).
By default, we use an SGD optimizer with momentum$=0.9$,
weight decay$=5e-4$. When using STLR, we set the maximum learning rate to $0.5$.
We also run experiments with two other model architectures (ResNet-18
and DenseNet-121) on \cifar{}. When training an initial model $F_0$, we train the model for 80 epochs using the same optimizer
settings and the STLR scheduler as the ResNet-9 experiments.

To update the model (i.e., from $F_{i-1}$ to $F_i$), our default updating setting is to fine-tune each model $F_{i-1}$
with the new training data $D_i$ for 5 epochs using an SGD optimizer with a constant learning rate$=0.01$, momentum$=0.9$ and
weight decay$=5e-4$. We set learning rate$=0.01$ since it produces the
best normal accuracy among our experiments.
In \S\ref{sec:training}, we also experiment with an SGD optimizer with STLR setting max learning rate$=0.5$, momentum$=0.9$ and
weight\_decay$=5e-4$. We vary the learning rates and training epochs in \S\ref{sec:training} and
report the normal accuracy if it degrades over $2\%$ compared to our initial setting.

\section{Additional Experimental Results.}
\label{sec:app_expr_results}

\vspace{-0.05in}
\subsection{Normal Accuracy Drop with Data Distribution Drifts.}
\label{subsec:acc_drop}
\vspace{-0.05in}

Figure~\ref{fig:acc_drop} shows how normal accuracy drops on a static model when
inference data distribution drifts over time. For \mnist{}, we rotate the test images
by $4^\circ$ each time, and for \cifar{}, we change the hue of test images by a factor
or 0.02 each time.

\vspace{-0.05in}
\subsection{Persistent Poisoning}
\label{subsec:persist}
\vspace{-0.05in}

We also consider powerful attackers who can continuously poison the training
data.  Clearly, if the attacker can poison a sufficient fraction of $D_i$,
$\forall i$, the backdoor will remain intact in the time-varying model (which
we validated empirically). Here, a more interesting question is: ``Does poisoning
more model updates make the backdoor survive longer once poison stops?''  Therefore,
we poison the model for more model updates (2, 5) and compare the backdoor survivability.
Figure~\ref{fig:pu} shows that in general poisoning more model updates will increase the
backdoor survivability, but there is no guarantee on this and the trend varies
with different attack methods on different datasets.
When using STLR, our results indicate that the backdoor survivability stays 0 for
persistent poisoning when the poison setting same as Figure~\ref{fig:pu}. The detailed results
can be found in Figure~\ref{fig:pu_stlr}.

\subsection{Number of Training Epochs.}
\label{subsec:app_te}
 Figure~\ref{fig:ne} shows that increasing the
 number of training epochs during model updates from 5 to 20 reduces the
 backdoor survivability in most cases. While we
 are multiplying the training efforts, the impact on the backdoor survivability is
 very limited. This is likely because although additional training epochs allows the
 model to learn more new data features (and thus forget existing features
 like the backdoor faster), the model weights might be trapped in the local
 minimum in DNNs. 

\vspace{-0.05in}
\subsection{Existing Defenses}
\label{subsec:defense}
\vspace{-0.05in}

We test \nc{} and \tabor{} with the three attacks (Badnets, Blend and Wanet) on \cifar{}
dataset. We first run the two defenses on clean and backdoored initial models ($F_0$s).
As shown in Table~\ref{tab:defense_f0_99}, when we set the attack success
rate threshold to 99\% as guided in the original papers~\cite{wang2019neural, guo2020towards},
both \nc{} and \tabor{} work well on all three attacks (both defenses suggest that
models with anomaly index over 2 are detected as backdoored). However, when we run
\nc{} and \tabor{} on the fine-tuned backdoored models after the poison stops, Table~\ref{tab:defense_f1_99}
shows that \nc{} and \tabor{}
fail to detect the backdoor models in most cases (only \tabor{} can detect the backdoored models
attacked by Blend).

A straightforward way of adapting the defenses is to decrease the attack success rate threshold to
a smaller threshold, like 75\%, 50\% or 25\%, since the attack success rate on fine-tuned models are lower than 99\%. The defender can set the threshold according to their need. However, our results imply that this direct adaptation does not work. Table~\ref{tab:defense_f0}
shows that when reducing the attack threshold, the detection results are unstable in two aspects:
1) the anomaly index for clean models may increase over 2 (attack success rate threshold as 50\% for
\nc{} and attack success rate threshold 25\% for \tabor{}); 2) the anomaly index for backdoored
model may drop below 2.